\documentclass[english]{extarticle}
\usepackage[T1]{fontenc}
\usepackage[utf8]{luainputenc}
\usepackage{geometry}
\usepackage{color}
\usepackage{float}
\usepackage{amsmath}
\usepackage{amsthm}
\usepackage{amssymb}
\usepackage{stmaryrd}
\usepackage{esint}
\PassOptionsToPackage{normalem}{ulem}
\usepackage{ulem}

\usepackage{graphicx}
\makeatletter

\floatstyle{ruled}
\newfloat{algorithm}{tbp}{loa}
\providecommand{\algorithmname}{Algorithm}
\floatname{algorithm}{\protect\algorithmname}

\theoremstyle{plain}
\newtheorem{thm}{\protect\theoremname}
\theoremstyle{plain}
\newtheorem{prop}[thm]{\protect\propositionname}
\theoremstyle{definition}
\newtheorem{defn}[thm]{\protect\definitionname}
\theoremstyle{remark}
\newtheorem{rem}[thm]{\protect\remarkname}
\theoremstyle{plain}
\newtheorem{cor}[thm]{\protect\corollaryname}
\theoremstyle{plain}
\newtheorem{lem}[thm]{\protect\lemmaname}
\ifx\proof\undefined
\newenvironment{proof}[1][\protect\proofname]{\par
\normalfont\topsep6\p@\@plus6\p@\relax
\trivlist
\itemindent\parindent
\item[\hskip\labelsep\scshape #1]\ignorespaces
}{%
\endtrivlist\@endpefalse
}
\providecommand{\proofname}{Proof}
\fi
\theoremstyle{remark}

\makeatother

\usepackage{babel}
\providecommand{\claimname}{Claim}
\providecommand{\definitionname}{Definition}
\providecommand{\lemmaname}{Lemma}
\providecommand{\propositionname}{Proposition}
\providecommand{\remarkname}{Remark}
\providecommand{\theoremname}{Theorem}
\providecommand{\corollaryname}{Corollary}

\newcommand{\V}{\mathcal{V}}

\usepackage{latexsym}
\begin{document}
\begin{center}
\textbf{\large  DISTRIBUTED STOCHASTIC APPROXIMATION\\ \ \\ WITH LOCAL PROJECTIONS}
\end{center}

\ \\

\begin{center}
SUHAIL MOHMAD SHAH AND VIVEK S.\ BORKAR\footnote{ Work supported in part by a grant for `Approximation of High Dimensional Optimization and Control Problems' from the Department of Science and Technology, Government of India.}\\
\ \\
Department of Electrical Engineering,\\
Indian Institute of Technology Bombay,\\
Powai, Mumbai 400076, India.\\
(suhailshah@ee.iitb.ac.com, borkar.vs@gmail.com).
\end{center}

\ \\

\noindent \textbf{Abstract} We propose a distributed version of a stochastic approximation scheme constrained to remain in the intersection of a finite family of convex sets. The projection to the intersection of these sets is also computed in a distributed manner and a `nonlinear gossip' mechanism is employed to blend the projection iterations with the stochastic approximation using multiple time scales.\\


\newpage

\section{Introduction}

In a landmark paper, Tsitsiklis et al \cite{Tsitsiklis} laid down a framework for distributed computation, notably for distributed optimization algorithms. They developed it further in \cite{Bertsekas}. There was a lot of subsequent activity and variations, an extensive account of which can be found in \cite{Sayed}. The key idea of \cite{Tsitsiklis} was to combine separate iterations by different processors/agents with an averaging mechanism that couples them and leads to a `consensus' among the separate processors, often on the desired objective (e.g., convergence to a common local minimum). The averaging mechanism itself has attracted much attention on its own as gossip algorithm for distributed averaging \cite{Shah} and various models for dynamic coordination \cite{Olfati}. One way to view these algorithms is as a two time scale dynamics with averaging on the fast or `natural' time scale dictated by the iterate count $n = 1,2,\cdots$ itself, with the rest (e.g., a stochastic gradient scheme) being a \textit{regular} perturbation thereof on a slower time scale, dictated by the chosen stepsize schedule. Then the convergence results can be viewed as
the fast averaging process leading to the confinement of the slow dynamics to the former's invariant subspace, which is the one dimensional space spanned by constant vectors. This is just a fancy way of interpreting consensus, but lends itself to some natural generalizations which were taken up in \cite{Mathkar}. Here the averaging was replaced by some nonlinear operation with the conclusion that it forced asymptotic confinement of the slow iterates given by a stochastic approximation scheme (SA for short) to \textit{its} invariant set. Our aim here is to leverage this viewpoint to propose a distributed algorithm for constrained computation wherein we want asymptotic confinement to the intersection of a finite family of compact convex sets. A prime example of such an exercise is constrained optimization, though the scheme we analyze covers the much broader class of projected stochastic approximation algorithms. In particular, the nonlinear operation gets identified with a distributed scheme for projection onto the intersection of a finite family of convex sets.\\

To contrast this with the classical projected stochastic approximation \cite{Kushner}, note that in the latter, a projection is performed at each step. In practice this may entail another iterative scheme to compute the projection as a subroutine, so that one waits for its near-convergence at each step. In our scheme, this iteration is embedded in the stochastic approximation iteration as a fast time scale component so that it can be carried out concurrently.\\

Our scheme is inspired by the results of \cite{Mathkar}. The results of \cite{Mathkar}, however, use strong regularity conditions such as  Frechet differentiability of the nonlinear map (among others) which are unavailable in the present case, making the proofs much harder. An additional complication is that the `fast' dynamics in one of the algorithms is itself a two time scale dynamics with stochastic approximation-like time-dependent iteration. This creates several additional difficulties. Thus the proofs of \textit{ibid.} cannot be applied here directly.  This is even more so for the test of stability  in section 4, which requires a significantly different proof.\\


\textbf{Relevant literature :} The literature on distributed algorithms is vast, mostly building upon the seminal work of \cite{Tsitsiklis}. The most relevant works for our purposes are \cite{Lee}-\cite{Srivastava1}, where distributed optimization algorithms with local constraints are considered. However, all of these are concerned primarily with convex optimization (and without noise). In \cite{Ozdaglar}, the convergence analysis is done for projected convex optimzation for the special case when the network is completely connected. The work in \cite{Srivastava}, \cite{Srivastava1} extends the algorithm of \cite{Ozdaglar} and its analysis to a more general setting including the presence of noisy links. \cite{Bianchi} considers distributed optimization for non-convex functions but again without distributed projection.\\

\textbf{Contributions :} The main contribution of this paper is that algorithms are provided for projected distributed SA where the projection component is distributed (so that only local constraints are required at the nodes). This is helpful when projection onto the entire constraint set is not possible (constraints are known only locally) or computationally prohibitive (a large number of constraints).  Moreover, in our algorithms the projection is tackled on a faster time scale which is a new feature that has not been explored previously and is of independent interest.
In both the algorithms proposed here, the projection component gives the exact projection of the point provided and not just a feasible point. This fact seems to be crucial in extending previous works, which mainly consider convex optimization with noiseless gradient measurements, to a fully distributed algorithm for the much more general case of a Robbins-Monro type stochastic approximation scheme. The feasibility and convergence properties of previous works depend critically on the specifics of convex optimization such as the convexity of the function being optimized. The main compromise here is that our  convergence results are established under a stability assumption which would not be required for a compact constraint set if exact projection was performed at every step.\\

The remainder of this section sets up the notation and describes our algorithm. The next section summarizes  the algorithm and some key results from projected dynamical systems. Section 3 details the main convergence proof assuming boundedness of iterates. The latter is separately proved in section 4. Section 5 provides some numerical results.\\

\subsection{Notation}
The projection
operator onto a constraint set $\mathcal{X}$ is denoted by $P_{\mathcal{X}}(\cdot)$,
i.e.
\[
\mathbf{}P_{\mathcal{X}}(y)=\textrm{arg}\min_{x\in\mathcal{X}}\|y-x\|
\]
is the Euclidean projection onto $\mathcal{X}$. We will be considering the case where
$$\mathcal{X}=\bigcap_{i=1}^{N}\mathcal{X}_{i}.$$
The projection operator
onto an individual constraint set $\mathcal{X}_i$ is denoted by $P^{i}(\cdot)$.

Since we are dealing with distributed computation, we use a stacked
vector notation. In particular $x_{k}=[(x_{k}^{1})^T\cdot\cdot\cdot (x_{k}^{N})^T]^{T}$
where $x_{k}^{i}$ is the value stored at the $i$'th node,
 so that
the superscript indicates the node while the subscript the iteration
count. Similarly any function under consideration denoted by $h(\cdot):\mathbb{R}^{Nn}\to\mathbb{R}^{Nn}$
represents $h(x_{k})=[h^{1}(x_{k}^{1})^T\cdot\cdot\cdot h^{n}(x_{k}^{N})^T]^{T}$
where each $h^{i}:\mathbb{R}^{n}\to\mathbb{R}^{n}$ is a component
function. We let a bold faced $\mathbf{P}(\cdot)$ denote the following
operator on the space $\mathbb{R}^{Nn}$ :
\[
\mathbf{P}(x_{k})=[P^{1}(x_{k}^{1})^T, \cdot\cdot\cdot, P^{N}(x_{k}^{N})^T]^{T}
\]
 with $P^{i}$ as described above. Let $\mathbf{1}$ denote the constant vector of all $1$'s of appropriate dimension. We let $\varotimes$ denote the
Kronecker product between two matrices and $\langle x_{k}\rangle$
denote the average of $x_{k}$, so that
\[
\langle x_{k}\rangle=\frac{1}{N}\mathbf{(1^{T}\varotimes}I_{n})x_{k}=\frac{x_{k}^{1}+\cdot\cdot\cdot+x_{k}^{N}}{N}
\]
with a bold faced $\mathbf{\langle x_{k}\rangle}=[\langle x_{k}\rangle^T,\cdot\cdot\cdot,\langle x_{k}\rangle^T]^T$.
A differential inclusion is denoted as
\[
\dot{x}\in F(x)
\]
where $F(\cdot)$ is a set valued map. Let $F^{\delta}( \cdot )$ denote the the following set :
\[
F^{\delta}(x)=\big\{ z\in\mathbb{R}^{n}\,:\,\exists x' \,\textrm{such that}\, \ \|x-x'\|<\delta,  \ \text{d}(z,F(x'))<\delta\big\}
\]
where $\text{d}(z,A)=\inf_{y \in A}\|z-y\|$ is the distance of a point $z$ from a set $A$. The notation $f(x)=o(g(x))$ is used to denote the fact
\[
\lim_{x\to\infty}\frac{f(x)}{g(x)}\to0,
\]
while $f(x)=\mathcal{O}(g(x))$ represents
\[
\limsup_{x\to\infty}\left|\frac{f(x)}{g(x)}\right| \leq M
\]
for some constant $M < \infty$.

\section{Background}

\subsection{Set-up}
Suppose we have a network of $N$ agents indexed by $1, ..., N.$ We associate with each agent $i$, a function $h^i : \mathbb{R}^n\to \mathbb{R}^n$ and a constraint set $\mathcal{X}_{i}$. The global constraint, given as the intersection of all $\mathcal{X}_{i}$'s,  is denoted by $\mathcal{X}$ :$$\mathcal{X}=\bigcap_{i=1}^n \mathcal{X}_i.$$

Also, let $H:\mathbb{R}^n\to\mathbb{R}^n$ denote
\begin{equation}\label{H-def.}
H(\cdot) := \frac{1}{N}\sum_{i=1}^Nh^i(\cdot).
\end{equation}
In many applications, one has $h^i(\cdot) = f(\cdot)\,\forall i$.
Let the communication network be modeled by a static undirected graph
$\mathcal{G=}\{\mathcal{V},\mathcal{E}\}$ where $\mathcal{V}=\{1,...,N\}$
is the node set and $\mathcal{E\subset\mathcal{V}\mathcal{\times}\mathcal{V}}$
is the set of links $(i,j)$ indicating that agent $j$ can send information
to agent $i$. All of the arguments presented here can be extended
to a time-varying graph under suitable assumptions as in \cite{Phade}. Here we deal only with
a static network for ease of notation.\\

We associate with the network a non-negative weight matrix
 $Q = [[q_{ij}]]_{i,j \in \V}$ such that
 $$q_{ij} > 0 \Longleftrightarrow (i,j)\in\mathcal{E}.$$
 In addition,
the following  assumptions are made on the matrix $Q$ and the constraint set :\\

\textbf{Assumption 1:}\\

i) {[}\textit{Convex constraints}{]} For all $i$, $\mathcal{X}_{i}$ are convex and compact. Also, the set $\mathcal{X}$ has a non-empty interior.\\

ii) {[}\textit{Double Stochasticity}{]} $\mathbf{1}^{T}Q=\mathbf{1}^{T}$
and $Q\mathbf{1}=\mathbf{1}$.\\

iii) {[}\textit{Irreducibility and aperiodicity}{]} We assume that the underlying graph is irreducible, i.e., there is a directed path from any  node to any other node, and  aperiodic, i.e., the g.c.d.\ of lengths of all paths from a node to itself is one. It is known that the choice of node in this definition is immaterial. This property can be guaranteed, e.g.,  by making $q_{ii}>0$ for some $i$.\\

This implies that the spectral norm $\gamma$ of $Q-\frac{\mathbf{1}\mathbf{1}^{T}}{N}$
satisfies $\gamma<1$.
This guarantees in particular that
\begin{equation}
\|\big(Q^{k}-Q^{*}\big)u\|\leq\kappa\beta^{-k}\|u\| \label{eq:-3}
\end{equation}
for some $\kappa > 0, \beta>1$, with $Q^{*}$ denoting the matrix $\frac{\mathbf{1}\mathbf{1}^{T}}{N}$.
\[
\]

\subsection{Distributed Projection Algorithms}

We now give details of two algorithms for computing a distributed projection. ``Distributed Projection'' here
means that the projection onto a particular constraint set $\mathcal{X}_i$ is performed by
just one processor/agent and they communicate information with each
other in order to compute projection onto the intersection  $\mathcal{X}.$\\

\noindent{\textit{A. Gradient Descent} :} The first approach involves viewing the projection problem as the minimization of the error
norm subject to the appropriate constraints, so that the projection of
a point $x_{0}$ can be thought of as the solution to the following optimization
problem :
\begin{equation}\label{optproj}
\min_{z \in \mathbb{R}^n}\,\,\|z-x_{o}\|^{2}
\end{equation}
\[
\text{s.t. }\,\,z\in\mathcal{X}=\cap \, \mathcal{X}_{i}\,,\,i=1,...,N
\]

In the distributed setting we associate each constraint $\mathcal{X}_{i}$ with an agent $i$. To solve the projection problem in a distributed fashion we first re-cast it as :

\[
\begin{array}{ccc}
\min_{z\in\mathbb{R}^{n}}\,\,\|z-x_{0}\|^{2}\,\,\,\,\,\,\,\,\,\,\, &  & \,\,\,\,\,\,\,\,\,\min_{\{z^i\in\mathbb{R}^{n},i=1, \cdots, N\}}\,\,\frac{1}{N}\sum_{i=1}^{N}\|z^{i}-x_{0}\|^{2}\\
\text{s.t. }\,\,z\in\mathcal{X}\,\,\,\,\,\,\,\,\,\,\, & \Leftrightarrow & \,\,\,\,\,\,\,\,\,\,\,\text{s.t. }\,\,z^{i}\in\mathcal{X}\, \ \forall\,i\\
 &  & \,\,\,\,\,\,\,\,\,\,\,z^{i}=z^{j}\, \ \forall\,i,j
\end{array}
\]

It is obvious that both the  problems have the same unique minimizer. The problem on the right can be solved by using a distributed gradient
descent of the form of equations 2a-2b, \cite{Lee} which for our case becomes :

\begin{equation}\label{graddesc}
z_{k+1}^{i}= P^i \Big\{ \sum_{j=1}^{N}q_{ij}z_{k}^{j}-b_{k}[\sum_{j=1}^{N}q_{ij}z^j_{k}-x_{0}] \Big\}
\end{equation}
where $b_k$ satisfies $ \sum_kb_k = \infty, \ \sum_k b_k^2 < \infty$. Note that the term inside the square brackets is proportional to the gradient of the function in (\ref{optproj}) evaluated at $\sum_{j=1}^{N}q_{ij}z^j_{k}$ .

\begin{lem} \label{projlem}
For any $x_{0}\in\mathbb{R}^{n}$ the iteration (\ref{graddesc}) converges
to the projection of $x_0$ upon $\mathcal{X}$, i.e.,
\[
z_{k}^{i}\to P_\mathcal{X}(x_{0})\,\,\forall\,i.
\]
\end{lem}
\begin{proof}

The function being optimized is strongly convex so that the optimal point is unique. We can directly invoke Prop.\ 1 \cite{Lee}
which guarantees convergence to the unique minimum in the set $\mathcal{X}$
and this unique minimum is the projection
point $P_\mathcal{X}(x_{0})$. (Note that Assumption 1 here is necessary for the equations 2a-2b of \cite{Lee} to converge to the projection point. Specifically Assumptions 1-5 of \cite{Lee} are satisfied for our case.)
\end{proof}

 \noindent{\textit{B. Distributed Boyle-Dykstra-Han} :} The algorithm originally proposed in \cite{Phade} is as follows:

\begin{algorithm}[H]
\textbf{Input :} $y_{0}\in\mathbb{R}^{n}$ ;\\

1: Set $z_{0}^{i}=y_{0}\,$ and $x_{0}^{i}=0$ for all $i$.\\

2: \textbf{for} k=1,2,.... \textbf{do}\\

3: At each node $i\in\mathcal{I}$ : \textbf{do }\\

3: $x_{k}^{i}=\sum_{j=1}^{N}q_{ij}\big\{ x_{k-1}^{j}+P^{j}(z_{k}^{j})\big\}-P^{i}(z_{k}^{i})$.\\

4: $z_{k+1}^{i}=z_{k}^{i}+b_{k}x_{k}^{i}$.\\

5: \textbf{end for }

\caption{Distributed Projection Scheme }
\end{algorithm}

The step size $b_k$ is assumed to satisfy : $$ \sum_kb_k = \infty, \ \sum_k b_k^2 < \infty$$

and for any $\epsilon>0$, there exists an $\alpha \in (1,1+\epsilon)$ and some $k_0$ such that
\begin{equation}\label{apndxb1}
\alpha b_{k+1}\geq b_k,\,\,\,\forall k>k_0,
\end{equation}

Let us write the above in vector notation as
\begin{eqnarray}
x_{k} & =& (Q\varotimes I_{n})\{x_{k-1}+\mathbf{P}(z_{k})\}-\mathbf{P}(z_{k}) \label{exeq}\\
z_{k+1} & =& z_{k}+b_{k}x_{k}. \label{zedeq}
\end{eqnarray}

The next theorem states that the above algorithm gives the exact projection
of the initial point $y_{0}$.
\begin{thm}\label{Sohamresult}
Suppose $z_{0}^{i}=y_{0}\,\forall\,i.$ Then $z_{k}\to z^{*}=[z_{1}^{*},...,z_{n}^{*}]$,
such that
\[
P^{i}(z_{i}^{*})=\mathbf{P}(y_{0}).
\]

\end{thm}
The detailed proof can be found in \cite{Phade} where an ODE approximation
 is used along with the associated Lyapunov function $z \mapsto \|z-z^{*}\|^{2}$.\\

 \begin{rem}\label{remark1}
The above algorithm is inspired from a parallel version of Boyle-Dykstra-Han originally proposed in \cite{Iusem}. The main difference from \cite{Iusem} is that the weights in \cite{Phade} are derived from the related graph of the communication network. The compromise is that a decaying time step is required to ensure convergence which may affect the convergence rate. We invite the reader to go through \cite{Phade} to get some more intuition regarding the algorithm and what the exact roles of $x_k$ and $z_k$ are in the context of Boyle-Dykstra-Han algorithm.
\end{rem}

\subsection{The algorithm}

The first algorithm for distributed projected SA is stated below along with the assumptions  on the various terms involved.

\begin{algorithm}[H]

1: Initialize $y_{0}^{i}$ and set $z_{k}^{i}=0$ for all
$i$.\\

2: \textbf{for} k=1,2,.... \textbf{do:}\\

3: At each node $i\in\mathcal{I}$ : \textbf{do }\\

a: [\textbf{Fast Time Scale]} Distributed Projection Step:\\

a1:  $z_{k+1}^{i}=P^{i}\big(\sum_{j=1}^{N}q_{ij}z_{k}^{j}-b_{k}(\sum_{j=1}^{N}q_{ij}z_{k}^{j}-y_{k}^{i})\big)$\\

b: Derive a noisy sample "$h^{i}(y_{k}^{i}) + M^i_{k+1}$" of  $h^{i}(y_{k}^{i})$ from a sampling oracle.\\

c: \textbf{[Slow Time Scale]} Distributed Stochastic Approximation Step:\\

c1: $y_{k+1}^{i}=\sum_{j=1}^{N}q_{ij}y_{k}^{j}+a_{k}(z^i_k-y_{k}^{i})+a_{k}(h^{i}(y_{k}^{i}) + M^i_{k+1})$\\

4: \textbf{end for }\\

\caption{Distributed SA with Gradient Descent (DSA-GD)}
\end{algorithm}

We make the following key assumptions :\\

\textbf{Assumption 2 :}\\

i) For each $i$, the function $h^i :\mathbb{R}^n \to \mathbb{R}^n$ is Lipschitz\\

ii) For each $i$, $\{M_k^i\}$ is a martingale difference sequence with respect to the the filtration  {$ \mathcal{F}^i_k := \sigma (y^i_{\ell},M^i_{\ell}, \ell \leq k)$},  i.e., it is a sequence of zero mean random variables satisfying:
$$\mathbb{E}\left[M^i_{k+1} | \mathcal{F}^i_k \right] = 0.$$ where $\mathbb{E}[\, \ \cdot \ | \ \cdot \ \,]$ denotes the conditional expectation. In addition we also assume a conditional variance bound
\begin{equation}
\mathbb{E}\left[\|M^i_{k+1}\|^2 | M^i_{\ell}, y^i_{\ell}, \ell \leq k\right] \leq K(1 +\|y^i_k\|^2) \ \ \forall k,i\text{ a.s.} \label{mgbd}
\end{equation}
for some constant $K>0$.\\

iii) Stepsizes $\{a_{k}\}$
and $\{b_{k}\}$ are positive scalars which satisfy :
$$a_{k}=o(b_{k}), \ \sum_ka_k = \sum_kb_k = \infty, \ \sum_k(a_k^2+b_k^2) < \infty,$$

iv) The iterates $y_k$ are a.s.\ bounded, i.e
\begin{equation}
 \sup_k \|y_k\|<\infty\ \ \mbox{a.s.} \label{y-stable}
\end{equation}
As in classical analysis of stochastic approximation algorithms by the o.d.e. method, we prove convergence assuming the stability condition Assumption 2(iv). In Section 4 we give sufficient conditions for  the latter to hold.

We now give a variant of the above algorithm using the distributed Boyle-Dykstra-Han algorithm (Algorithm 1) :

\begin{algorithm}[H]

1: Initialize $y_{0}^{i}$ and set $z_{k}^{i}=0,\,x_{k}^{i}=0$ for all
$i$\\

2: \textbf{for} k=1,2,.... \textbf{do:}\\

3: At each node $i\in\mathcal{I}$ : \textbf{do }\\

a: [\textbf{Fast Time Scale]} Distributed Projection Step:\\

a1: $x_{k}^{i}=\sum_{j=1}^{N}q_{ij}\{x_{k-1}^{j}+P^{j}(z_{k}^{j}+y_{k}^{j})\}-P^{i}(z_{k}^{i}+y_{k}^{i})$\\

a2: $z_{k+1}^{i}=z_{k}^{i}+b_{k}x_{k}^{i}$\\

b: Derive a noisy sample "$h^{i}(y_{k}^{i}) + M^i_{k+1}$" of  $h^{i}(y_{k}^{i})$ from a sampling oracle.\\

c: [\textbf{Slow Time Scale]} Distributed Stochastic Approximation Step:\\

c1: $\bar{y}_{k}^{i}=P^{i}(y_{k}^{i}+z_{k}^{i})$\\

c2: $y_{k+1}^{i}=\sum_{j=1}^{N}q_{ij}y_{k}^{j}+a_{k}(\bar{y}_{k}^{i}-y_{k}^{i})+a_{k}(h^{i}(y_{k}^{i}) + M^i_{k+1})$\\

4: \textbf{end for }\\

\caption{Distributed SA with Boyle-Dykstra-Han (DSA-BDH)}
\end{algorithm}

\begin{rem}\label{remark1}
The assumption of compact $\mathcal{X}_i$ is not necessary for Algorithm 2. As long as the iterates are assumed to be bounded (both $y_k$ and $z_k$), `compact sets' can be replaced by `closed sets'.\\
\end{rem}

\begin{rem}\label{remark1}
For the various terms involved in Algorithm 3, Assumption 2 continues to apply. The only addition is that Assumption 2(iii) now includes the condition (\ref{apndxb1}).\\
\end{rem}

\begin{rem}\label{remark1}
Both algorithms 2 and 3 have the same objective, i.e. distributed projected SA. For the rest of the paper we refer to Algorithm 2 as DSA-GD and Algorithm 3 as DSA-BDH.
\end{rem}

\subsection{Projected dynamical systems}

We analyze the algorithms  using the ODE (for `Ordinary Differential Equations') approach for analyzing stochastic approximation, extended to differential inclusions \cite{Benaim}. Define the normal cone $N_{\mathcal{X}}(x)$ to be the set
of outward normals at any point $x\in\partial\mathcal{X}$ where $\partial\mathcal{X}$ is the boundary of the set $\mathcal{X}$, i.e.,
\begin{equation}
N_{\mathcal{X}}(x)\doteq\{\gamma\in\mathbb{R}^{n}: \langle\gamma,x-y\rangle\geq0\, \ \forall\, \ y\in\mathcal{X}\} \label{cone}
\end{equation}
with $N_{\mathcal{X}}(x)=\{0\}$ for any point $x$ in the interior of $\mathcal{X}$. The relevant differential inclusion for our problem is
\begin{equation}
\dot{x}\in H(x(t))-N_{\mathcal{X}}(x(t)), \label{inclusion2}
\end{equation}
\[
x(t)\in\mathcal{X}\, \ \ \forall t\in[0,T]
\]
where $H(\cdot)$ is as in (\ref{H-def.}). This inclusion is identical to the well known "Projected Dynamical System" considered in \cite{Nagurney} :

\begin{equation*}
\dot{x}=\Pi(x,h(x))
\end{equation*}
where $\Pi(x,h(x))$ is defined to be the following limit for any $x\in\mathcal{X}$ :
\[
\Pi(x,h(x))\doteq\lim_{\delta\to0}\frac{P_{\mathcal{X}}(x+\delta h(x))-x}{\delta}
\]
The proof of the fact that the operator $\Pi(\cdot,\cdot)$ is identical to the RHS of (\ref{inclusion2}) is provided in (\cite{Dupuis}, Lemma 4.6).
The following theorem  is borrowed from (\cite{Brogliato}, Corollary 2) and (\cite{Cojocari}, Theorem 3.1).\\

\begin{thm} \label{BrogCoj}
For a convex $\mathcal{X}$, (\ref{inclusion2}) is  well posed, i.e. a unique solution exists.
\end{thm}

We recall the following notion  from \cite{Benaim} where more general differential
inclusions are considered:\\

\begin{defn}(\cite{Benaim})
Suppose $F$ is a closed set valued map such that $F(x)$ is a non-empty compact convex
set for each $x$. Then a perturbed solution
$y$ to the differential inclusion
\begin{equation}
\dot{x}\in F(x) \label{DifInc}
\end{equation}
is an absolutely continuous function which satisfies:\\

i) $\exists$ a locally integrable function $t\to U(t)$ such that
for any $T>0$, $$\lim_{t\to\infty}\sup_{0\leq v\leq T}\big|\int_{t}^{t+v}U(s)ds\big|=0,$$

ii) $\exists$ a function $\delta:[0,\infty)\to[0,\infty)$ with $\delta(t)\to0$ as $t \to \infty$
such that
\[
\dot{y}-U(t)\in F^{\delta(t)}(y).
\]

\end{defn}

There is no guarantee that the perturbed solution remains close to a
solution of (\ref{DifInc}),  however, the following assumption  helps in establishing some form of convergence. Let $\Lambda$ denote the equilibrium set of (\ref{inclusion2}), assumed to be non-empty. Then :

\textit{
\[
\Lambda \subset \{x\,:\,H(x)\in N_{\mathcal{X}}(x)\}.
\]
}
\\

\noindent \textbf{Assumption 3 :}\textit{ There exists a Lyapunov function for
the set $\Lambda$, i.e.,  a continuously differentiable function $V:\mathbb{R}^{n}\to\mathbb{R}$
such that any solution $x$ to (\ref{inclusion2}) satisfies
\[
V(x(t))\leq V(x(0)) \ \forall \\ t > 0
\]
 and the inequality is strict whenever $x(0)\notin \Lambda$.}\\

 If this assumption does not hold, the asymptotic behavior is a bit more complex. Specifically, under reasonable assumptions, an SA scheme converges a.s.\ to an invariant internally chain transitive set of the limiting o.d.e. This behavior extends to differential inclusions as well (see \cite{Benaim} or Chapter 5, \cite{BorkarBook}) and is what we would expect for our algorithm if we remove Assumption 3. \\

For instance if $H=-\nabla g$ for some continuously differentiable $g$, then
the above set represents the KKT points and the function $g$ itself
will serve as a Lyapunov function for the set $\Lambda$.\\

 The following result is from (\cite{Benaim}, Prop. 3.27) :
\begin{prop} \label{Perturbed}
Let $y$ be a bounded perturbed solution to (\ref{inclusion2}) and there exist a Lyapunov
function for a set $\Lambda$ with $V(\Lambda)$ having an empty interior.
Then
\[
\bigcap_{t\geq0}\overline{y\big([t,\infty)\big)}\subset\Lambda
\]
\end{prop}
\begin{rem}\label{remark1}
To prove the main result, we shall show that the suitably interpolated iterates generated
by the algorithm form a perturbed solution to the differential inclusion
(\ref{inclusion2}), so that using Theorem \ref{BrogCoj} and Proposition \ref{Perturbed} (along with Assumption
3), the algorithm is shown to converge to its
equilibrium set.\\

\end{rem}

To conclude this section, we provide some intuition behind the proposed algorithms. The SA part in DSA-GD is
 $$y_{k+1}^{i}=\sum_{j=1}^{N}q_{ij}y_{k}^{j}+a_{k}(z^i_k - y_{k}^{i})+a_{k}(h^{i}(y_{k}^{i}) + M^i_{k+1})$$

If the term $z^i_k$ asymptotically tracks the projection $P_\mathcal{X}(y^i_k)$ (see Lemma \ref{projlem1}), then this can be written as (modulo some asymptotically vanishing error)  :

 $$y_{k+1}^{i} \approx \sum_{j=1}^{N}q_{ij}y_{k}^{j}+a_{k}(h^{i}(y_{k}^{i}) + \underbrace{P_\mathcal{X}(y^i_k)-y_{k}^{i}}_{\text{projection error term}}+ M^i_{k+1})$$

The projection error term strives to keep the iterates inside the constraint set (compare this to the inclusion (\ref{inclusion2}) where the constraining term belongs to the normal cone). So although the above iteration may appear to behave like an unconstrained one, it achieves the same asymptotic behavior as what one would get by projecting onto the entire constraint set $\mathcal{X}$ at each step. An analogous intuition applies to DSA-BDH.

\section{Convergence proof}

We first deal with DSA-GD, the proof details for DSA-BDH are nearly the same and we give a brief outline in the second part of this section. Assumptions 1,2 and 3 are assumed to hold throughout this section.\\

\noindent \textit{A. Convergence of DSA-GD:} We rewrite some of the main steps in DSA-GD with a stacked
vector notation :
\begin{eqnarray}
z_{k+1} &=& \mathbf{P\big\{}(Q\otimes I_{N})z_{k}-b_{k}((Q\otimes I_{N})z_{k}-y_{k})\big\}, \label{medium} \\
y_{k+1} &=& (Q\varotimes I_{n})y_{k}+a_{k}(z_k-y_{k})+a_{k}(h(y_{k})+M_{k+1}). \label{slow}
\end{eqnarray}
Our main convergence result is :
\begin{thm}\label{main thrm}
Under Assumptions 1-3, we have almost surely
\[
y_{k}\to\{\mathbf{1\varotimes}y\,:\,y\in\Lambda\}\,\,\, .
\]
\end{thm}

The above theorem states that the variables $\{y_{k}^{i},\,i=1,...,N\}$
achieve consensus as expected and moreover, their limit points lie
in the equilibrium set $\Lambda$ of (\ref{inclusion2}). For distributed
optimization, this set corresponds to the KKT points of the related minimization problem.\\

\textit{Outline of the analysis}: To analyze the above algorithm,
 we proceed in three steps:
\begin{enumerate}
\item \textit{Consensus} : We first show that the iterates $y_{k}^{i}$ achieve consensus (Lemma \ref{consensus}). This will help us analyze the algorithm by studying it in
the average sense at each node, i.e., with $y_k^i \approx \langle y_{k}\rangle := 1/N\sum_{i=1}^{N}y_{k}^{i}$.
\item \textit{Feasibility } : Next, the condition $a_{k}=o(b_{k})$ is exploited to do a two time
scale analysis. Specifically, in the distributed projection scheme operating on a fast time scale, the variable $y_{k}$, evolving on a slow time scale,
is quasi-static, i.e., treated as a constant.
In turn the (slow) stochastic approximation step sees the faster iterations  (\ref{medium}) as quasi-equilibrated. Therefore it asymptotically
behaves as its projected version where the projection is taken to
be upon the entire set at each node rather than only its own particular
constraint set.
\item \textit{Convergence }: Finally, using the the above analysis, the slow iterates are shown to track the desired projected
dynamical system.
\end{enumerate}

We first show that the above algorithm achieves consensus, i.e.,
as $k\to\infty,\,\|y_{k}^{i} - \langle y_{k}\rangle\| \to 0.$\\

\begin{lem} \label{consensus}
$\lim_{k}\max_{i,j=1,..N}\|y_{k}^{i}-y_{k}^{j}\|=0\,\,\textrm{a.s}$.
Also,
\[
\|y_{k}^{i}-\langle y_{k}\rangle\|\to0\: \ \ \forall\:i.
\]
\end{lem}
\begin{proof}
Set
\[
Y_{k}= z_k-y_{k}+h(y_{k})+M_{k+1},
\]
 so that we can write (\ref{slow}) as
\[
y_{k+1}=(Q\varotimes I_{n})y_{k}+a_{k}Y_{k}.
\]
We have, for a fixed $k\leq n$ and any large $n$,
\begin{eqnarray*}
y_{n+k} &=& (Q\varotimes I_{n})y_{n+k-1}+a_{n+k-1}Y_{n+k-1} \\
&=& (Q^{2}\varotimes I_{n})y_{n+k-2}+a_{n+k-2}(Q\varotimes I_{n})Y_{n+k-2} \\
&& + \ a_{n+k-1}Y_{n+k-1},
\end{eqnarray*}
where we have used the fact that $(Q\varotimes I_{n})(Q\varotimes I_{n})=(Q^{2}\varotimes I_{n})$.
Iterating the above equation further, we have
\[
y_{n+k}=(Q^{k}\varotimes I_{n})y_{n}+a_{n}(Q^{k-1}\varotimes I_{n})Y_{n}+\cdots,
\]

\[
\cdots +a_{n+k-2}(Q\varotimes I_{n})Y_{n+k-2}+a_{n+k-1}Y_{n+k-1}
\]
i.e.,
\begin{equation}
y_{n+k}=(Q^{k}\varotimes I_{n})y_{n}+\{\Gamma(Y_{n+k-1},...,Y_{n})\} \label{Gamma}
\end{equation}
where $\Gamma(\cdot)$ is some linear combination of its arguments. In view of Assumption 2(iv) ($y_k$ is bounded),
\[
\|Y_{k}\|\leq M \ \ \ \ \mbox{w.p.} \ 1
\]
for some \textit{random} $M < \infty$. So we have,

\begin{align*}
 \|\Gamma(Y_{n+k-1},...,Y_{n})\|  & =\|\sum_{i=n}^{n+k-1}a_i (Q^{n+k-1-i} \varotimes I_{n})Y_{i} \| \\
 & \leq M \big(\sum_{i=n}^{n+k-1}a_i \big)\qquad \big( \because \|(Q^{n+k-1-i} \varotimes I_{n})\|=1 \big)\\
 & = \mathcal{O}(\sum_{i=n}^{n+k-1}a_i)
\end{align*}

Subtracting $(Q^{*}\varotimes I_{n})y_{n}$ from both sides in
(\ref{Gamma}), we get
\begin{equation}
y_{n+k}-(Q^{*}\varotimes I_{n})y_{n}=[(Q^{k}-Q^{*})\varotimes I_{n}]y_{n}+\{\Gamma(Y_{n+k-1},...,Y_{n})\} \label{difference}
\end{equation}

Using equation (\ref{eq:-3}) and taking norms in (\ref{difference}), we have :
\[
\|y_{n+k}-(Q^{*}\varotimes I_{n})y_{n}\|=\mathcal{O}(\beta^{-k})+\mathcal{O}(\sum_{i=n}^{n+k-1}a_i).
\]
Letting $n \to \infty$ followed by $k \to \infty$, it follows that  any limit
point $y_{*}$ of the sequence $\{y_{k}\}$ satisfies
\begin{equation}
y_{*}=(Q^{*}\varotimes I_{n})y_{*}.
\end{equation}
That is, $y_{*}^{i}=\frac{1}{N}\sum_{j=1}^{N}y_{*}^{j}$ for
any $i$, so that consensus is achieved and the consensus value is the average
of all the node estimates.
\end{proof}

We next argue that the algorithm
can be regarded as a two time scale iteration so that while analyzing the
behavior of $z_{k}$ (fast  variable), $y_{k}$ (slow  variable) can be regarded as a constant (cf. \cite{Borkar2time} or \cite{BorkarBook}, Chapter 6). We have
\begin{eqnarray}
z_{k+1} & =& \mathbf{P\big\{}(Q\otimes I_{N})z_{k}-b_{k}((Q\otimes I_{N})z_{k}-y_{k})\big\}\nonumber \\
 & =& \mathbf{P\big\{}Q\otimes I_{N})z_{k}+b_{k}(\mu(z_{k},y_{k}))\big\} \label{fast2}
\end{eqnarray}
for a suitably defined $\mu$. The slow time scale iteration here is:
\begin{eqnarray}
y_{k+1} & =& (Q\varotimes I_{N})y_{k}+a_{k}(z_{k}-y_{k}+h(y_{k})+M_{k+1}) \nonumber \\
 & =& (Q\varotimes I_{N})y_{k}+a_{k}(\nu(z_{k},y_{k}) + M_{k+1}) \label{slow2}
\end{eqnarray}
for a suitably defined $\nu$. Since $a_{k}=o(b_{k})$ in (\ref{fast2})-(\ref{slow2}), the above pair of equations form a two time scale iteration (\cite{Borkar2time} or \cite{BorkarBook}, Chapter 6). So while analyzing (\ref{fast2}), we can assume that $y_{k}$ is $\approx$ a constant, say $\langle \mathbf{y} \rangle := [y,....,y]$ (we  take the same value at all the nodes because of consensus proved in Lemma \ref{consensus}), so that (\ref{fast2}) becomes :
\begin{equation}\label{neweq}
z_{k+1} = \mathbf{P\big\{}(Q\otimes I_{N})z_{k}-b_{k}((Q\otimes I_{N})z_{k}- \langle\mathbf{ y} \rangle)\big\}.
\end{equation}
This can be viewed as iteration (\ref{graddesc}) for the problem (\ref{optproj}) with $x_0=\langle\mathbf{ y} \rangle$. So Lemma \ref{projlem} implies $z_k^i\to P_{\mathcal{X}}(\langle y \rangle)$.

\begin{lem}\label{projlem1}
If $z^{*}(\langle \mathbf{y} \rangle)$ is the limit point of (\ref{neweq}) for any $\langle \mathbf{y} \rangle$, then $z^{*}(\langle \mathbf{y} \rangle )= P_{\mathcal{X}}(\langle \mathbf{y} \rangle)$. Also for all $i$,
\begin{equation*}
\|z_k^i - P_{\mathcal{X}} ( \langle y_k \rangle)\| \to 0.
\end{equation*}

\end{lem}

\begin{proof}
This first statement follows directly from the above discussion. The second is a direct consequence of Lemma 1, Chapter 6, \cite{BorkarBook} .
\end{proof}

We now prove the main result whose proof uses the
techniques of \cite{Benaim}, \cite{Bianchi}.\\

\noindent \textbf{Proof of Theorem \ref{main thrm} :}

\begin{proof}
The consensus part was already proved in Lemma \ref{consensus}. Multiplying both sides of (\ref{slow})  by $\frac{1}{N}(\mathbf{1}^{T}\otimes I_{n}$),
and using the double stochasticity of $Q$, we have
\[
\langle y_{k+1}\rangle=(1-a_{k})\langle y_{k}\rangle+a_{k}\langle z_k \rangle+a_{k}(\langle h(y_{k})\rangle+\langle M_{k+1}\rangle)
\]

We used in the above the fact
\begin{align*}
\frac{1}{N}(\mathbf{1}^{T}\otimes I_{n})(Q\otimes I_{n})y_{k} & =\frac{1}{N}(\mathbf{1}^{T}\otimes I_{n})y_{k}\\
 & =\langle y_{k}\rangle=\frac{1}{N}(y_{k}^{1}+....+y_{k}^{N}).
\end{align*}
Adding and subtracting $a_{k}P_{\mathcal{X}}(\langle y_{k}\rangle)$ on the right hand side,
\begin{equation}
\langle y_{k+1}\rangle=(1-a_{k})\langle y_{k}\rangle+a_{k}P_{\mathcal{X}}(\langle y_{k}\rangle)+a_{k}(\langle z_k \rangle-P_{\mathcal{X}}(\langle y_{k}\rangle)+\langle h(y_{k})\rangle+\langle M_{k+1}\rangle) \label{av-y}
\end{equation}

Note that,

$$\langle y_k \rangle -P_{\mathcal{X}}(\langle y_k \rangle ) \in \mathcal{N}_{\mathcal{X}}( P_{\mathcal{X}}(\langle y_{k}\rangle)) $$

Set $Z_{k}=\langle z_k \rangle -P_{\mathcal{X}}(\langle y_{k}\rangle)$  and $p_k = \langle h(y_{k})\rangle-H(P_{\mathcal{X}}(\langle y_{k}\rangle))$. So (\ref{av-y}) becomes :

\begin{equation}
\langle y_{k+1}\rangle \in  \langle y_{k}\rangle+a_{k}\Big( H(P_{\mathcal{X}}(\langle y_k \rangle )) - \mathcal{N}_{\mathcal{X}}( P_{\mathcal{X}}(\langle y_{k}\rangle)) + p_k + Z_k+\langle M_{k+1}\rangle \Big) \label{av-y-2}
\end{equation}\\

Let $t_{0}=0$ and $t_{k}=\sum_{i=0}^{k}a_{i}$ for any $k\geq1$, so that
 $t_{k}-t_{k-1}=a_{k-1}$. Define the interpolated trajectory $\Theta:[0,\infty)\to\mathbb{R}^{n}$
as :
\[
\Theta(t)=\langle y_{k}\rangle+(t-t_{k})\,\frac{\langle y_{k+1}\rangle-\langle y_{k}\rangle}{t_{k+1}-t_{k}}, \ t \in [t_{k}, t_{k+1}], \ k \geq 1
\]
By differentiating the above  we have

\begin{align*}
\frac{d\Theta(t)}{dt}= & \frac{\langle y_{k+1}\rangle-\langle y_{k}\rangle}{a_{k}}\,\,\,\,\forall\,t\in[t_{k},t_{k+1}]
\end{align*}
where we use the right, resp., left  derivative at the end points. Now define the following set valued map
\begin{equation}
F(x)=\{H(x)-W\,:\,W\in\mathcal{N}(x)\,, \|W\| \leq K\} \label{eff}
\end{equation}
where $0 < K < \infty$ is a suitable constant.\\

We get from (\ref{av-y-2}) using (\ref{eff}),
\begin{equation}
\frac{d\Theta(t)}{dt} \in F(P_{\mathcal{X}}(\langle y_{k}\rangle))+p_{k}+Z_{k}+\langle M_{k+1}\rangle, \ t \in [t_k, t_{k+1}]. \label{theta1}
\end{equation}

To finish the proof, $\Theta(\cdot)$ is first shown to be a perturbed
solution of (\ref{inclusion2}). Let $\eta(t)\doteq\|p_{k}\|+\|\Theta(t)-P_{\mathcal{X}}(\langle y_{k}\rangle)\|,\,t\in[t_{k},t_{k+1}), k \geq 0$.
Also define
$$U(t)= U_k := Z_k + \langle M_{k+1}\rangle,  $$
 for $t \in [t_k, t_{k+1}), k \geq 0$.
Then we have
\begin{equation}
\frac{d\Theta(t)}{dt}-U(t)\in F^{\eta(t)}(\Theta(t)). \label{incl}
\end{equation}
We have used the following fact in the above : for any set valued
map $F(\cdot)$ we have
\[
\forall\,(x,\hat{x})\ \in\mathbb{R}^{n}\times\mathbb{R}^{n},\;\!\ p+F(x)\subset F^{\|p\|+\|x-\hat{x}\|}(\hat{x}).
\]
If we show that $\sum_{k}a_{k}U_k<\infty$ and $\eta(t)\to0$, then by (\ref{incl})  $\Theta(\cdot)$ can be interpreted as a perturbed solution of the  differential inclusion
$$\dot{\Psi}(t) \in F(\Psi(t)).$$
Convergence to the
set $\Lambda$  then follows by Assumption 3 and the proof is complete.\\

We first prove that $\sum_{k}a_{k}U_k<\infty$, implying :
$$\lim_{t\to\infty}\sup_{0\leq v\leq T}\big|\int_{t}^{t+v}U(s)ds\big|=0.$$
for any $T>0$. We have
\[
U(t)=a_{k}\big[\underbrace{Z_{k}}_{I}+\underbrace{M_{k+1}}_{II}\big]
\]
for $t \in [t_k, t_{k+1})$. We consider the contributions of the two terms separately.\\

I:  We know from Lemma \ref{projlem1} that the mapping $ y \to P_{\mathcal{X}}(y)$ maps $y$ to the $y$-dependent limit point of (\ref{neweq}).
Then we have :
\begin{align*}
\|\langle z_{k} \rangle-P_{\mathcal{X}}(\langle y_{k}\rangle)\| &  =\|\langle z_{k} \rangle -\langle P_{\mathcal{X}}(\langle y_{k}\rangle) \rangle \|\,\,\,\,\,\big[\mbox{because} \ \ P_{\mathcal{X}}(\langle y_{k}\rangle)=\langle P_{\mathcal{X}}(\langle y_{k}\rangle) \rangle  ]\\
 & = \|\frac{1}{N}\sum_{i=1}^{N}\big( z^i_k -  P_{\mathcal{X}}(\langle y_{k}\rangle)    \big) \|\\
 & \leq\frac{1}{N}\sum_{i=1}^{N}\|z^i_{k}-P_{\mathcal{X}}(\langle y_{k}\rangle)\|\,\,\,\,\big[\text{Jensen's Inequality}\big]\\
 & \to0\,\,\,\,\,\,\big[\text{from Lemma \ref{projlem1}}\big].
\end{align*}
For $T > 0$, let $m(n) := \min\{k \geq n : \sum_{j=n}^{n+k}a_j \geq T\}, n \geq 0$. Then
\[
\sup_{\ell \leq m(n)}\sum_{k=n}^{\ell}a_{k}\|\langle z_k \rangle-P_{\mathcal{X}}(\langle y_{k}\rangle)\|\to0\,\,\text{as }n\to\infty\text{}.
\]

II: This term is the error induced by the noise. Note that the process $\sum_{m=0}^{k-1}a_mM^i_{m+1}, k \geq 1,$ is a zero mean square integrable martingale w.r.t.\ the increasing $\sigma$-fields  $\mathcal{F}^i_k := \sigma(M^i_m, y^i_m , m \leq k), k \geq 1$, with   $\sum_{m=0}^{\infty}a_m^2\mathbb{E}\left[\|M^i_{m+1}\|^2 | \mathcal{F}^i_m\right] < \infty$ by (\ref{mgbd}) and (\ref{y-stable}), along with the square-summability of $\{a_m\}$.  It follows from  the martingale
convergence theorem (Appendix C, \cite{BorkarBook}), that this martingale converges a.s. Therefore
\[
\sup_{\ell \leq m(n)}\|\sum_{k=n}^{\ell}a_{k}M^i_{k+1}\|\to0\,\,\,\text{a.s. }\forall i
\]

The claim follows for $\langle M_{k+1}\rangle$.\\

To prove that $\eta(t) \to 0$, consider for $t \in [t_k, t_{k+1}),\,k\geq 0$.

\begin{align*}
\eta(t) =\|p_k\|+\|\Theta(t)-P_{\mathcal{X}}(\langle y_{k}\rangle)\| & = \| \langle h(y_{k})\rangle-H(P_{\mathcal{X}}(\langle y_{k}\rangle))\| +\|\langle y_{k}\rangle - P_{\mathcal{X}}(\langle y_{k}\rangle)  \\
& \qquad \qquad \qquad + (\langle y_{k+1}\rangle-\langle y_{k}\rangle)\Big(\frac{t-t_{k}}{t_{k+1}-t_{k}}\Big) \| \\
& \leq\frac{1}{N}\sum_{i=1}^{N}\|h^{i}(y_{k}^{i})-h^{i}(P_{\mathcal{X}}(\langle y_{k}\rangle))\|+ \|\langle y_{k}\rangle - P_{\mathcal{X}}(\langle y_{k}\rangle)\| \\
 & \qquad \qquad \qquad + \|\langle y_{k+1}\rangle-\langle y_{k}\rangle\|\frac{\|t-t_{k}\|}{\|t_{k+1}-t_{k}\|}\\
& \leq\frac{C}{N}\sum_{i=1}^{N}\|y_{k}^{i}-P_{\mathcal{X}}(\langle y_{k}\rangle)\| +  \|\langle y_{k}\rangle - P_{\mathcal{X}}(\langle y_{k}\rangle)\| + \mathcal{O}(a_k) \\
\end{align*}
where $C > 0$ is a common Lipschitz constant for the $h^i$'s. Note that if we prove that $\| P_{\mathcal{X}}(\langle y_{k}\rangle)-\langle y_{k}\rangle\|\to 0$, all the terms in the above inequality go to zero as $k\uparrow \infty$. So to finish the proof, we prove this fact :\\

\noindent \textit{Claim :}
$\lim_k \inf_{x\in\mathcal{X}}\|y_{k}^{i}-x\| = 0\, \ \ \ \forall\,i,\,\text{as }k \uparrow \infty$\\

\noindent \textit{Proof.}  Let us first consider the following fixed point iteration :
\begin{displaymath}
\tilde{y}_{k+1}=(1-a_{k})\tilde{y}_{k}+a_{k}P_{\mathcal{X}}(\tilde{y}_{k}).
\end{displaymath}
By the arguments of \cite{BorkarBook}, Chapter 2, this has the same asymptotic behavior as the o.d.e.
\[
\dot{\bar{y}}=P_{\mathcal{X}}(\bar{y})-\bar{y}.
\]
Consider the Lyapunov function $V(\bar{y})=\frac{1}{2}\|\bar{y}\|^{2}$.
Then
\begin{align*}
\frac{d}{dt}V(\bar{y}(t)) & =\bar{y}(t){}^{T}(P_{\mathcal{\mathcal{X}}}(\bar{y}(t))-\bar{y}(t))\\
 & =\bar{y}(t){}^{T}P_{\mathcal{X}}(\bar{y}(t))-\|\bar{y}(t)\|^{2}.
\end{align*}
For any $v\in\mathbb{R}^{n}$,
the non-expansive property of $P_{\mathcal{X}}$ leads to
\[
\|P_{\mathcal{\mathcal{X}}}(v)\|\leq\|v\|.
\]
By the Cauchy-Schwartz inequality,
\[
\frac{d}{dt}V(\bar{y}(t))\leq 0.
\]
By Lasalle's invariance principle we have any trajectory $\bar{y}(\cdot)$
converge to the largest invariant set where $\frac{d}{dt}V(\bar{y}(t))=0$,
which is precisely the set $\mathcal{X}$. The claim for $\{\tilde{y}_k\}$ now follows by a standard argument as in Lemma 1 and Theorem 2, pp.\ 12-16, \cite{BorkarBook}.\\

Next, define maps $f_{k}$ by
\[
f_{k}(y)=(1-a_{k})y+a_{k}P_{\mathcal{X}}(y).
\]
Then (\ref{av-y}) becomes
\begin{equation}
\langle y_{k+1} \rangle= f_{k}(\langle y_{k} \rangle)+a_{k}(\langle h(y_{k})\rangle + M_{k+1}+\epsilon_k). \label{eq:-1-1}
\end{equation}
with $\epsilon_k = \langle z_k \rangle-P_{\mathcal{X}}(\langle y_{k} \rangle)$. From the preceding discussion, we have the following: for any fixed $k \geq 0$,
\begin{equation}\label{f-convg}
\lim_{m \uparrow \infty}F_{k,m}(\cdot) \ \Big(:= f_{m+k}\circ\cdot\cdot\circ f_{k}(\cdot)\Big) \to P_\mathcal{X}(\cdot).
\end{equation}
The  family $\{F_{k,m}\}$ of functions is non-expansive, therefore equi-continuous, and bounded (because $\{y_k\}$ is assumed to be bounded), hence relatively sequentially compact in $C(\mathbb{R}^n)$ by the Arzela-Ascoli theorem. Hence the above convergence is uniform on compacts, uniformly in $k$. Consider any convergent subsequence of $\{y_k\}$ with limit (say) $y^*$ and by abuse of notation, index it by $\{k\}$ again. From (\ref{eq:-1-1}), for any $k,m$ we have,
\begin{equation}
\langle y_{m+k+1} \rangle= f_{m+k}(\langle y_{m+k} \rangle)+a_{m+k}(\langle h(y_{m+k})\rangle + M_{m+k+1}+\epsilon_{m+k}). \label{bracketeq}
\end{equation}
We also have,
\begin{align*}
\| f_{m+k}(\langle y_{m+k} \rangle) -f_{m+k}\circ f_{m+k-1}(\langle y_{m+k-1} \rangle)\| & \leq \| \langle y_{m+k} \rangle - f_{m+k-1}(\langle y_{m+k-1} \rangle) \| \\
 & \qquad \qquad (\because \, f_{m+k} \text{ is non-expansive.})\\
 & = \ \| \langle y_{m+k} \rangle  - \langle y_{m+k-1} \rangle +\\
 & \qquad \qquad a_{m+k-1} \big( \langle y_{m+k-1} \rangle - P_\mathcal{X}(\langle y_{m+k-1} \rangle) \big) \| \\
 & = \mathcal{O}(a_{m+k-1}) \\
 & = o(1).
\end{align*}
 By iterating, we get :
$$\|f_{m+k}(\langle y_{m+k} \rangle) -  f_{m+k}\circ\cdot\cdot\circ f_{k}( \langle y_{k} \rangle )\| = o(1).$$
Combining this with (\ref{bracketeq}), we have
$$\|\langle y_{m+k+1}\rangle -  f_{m+k}\circ\cdot\cdot\circ f_{k}( \langle y_{k} \rangle )\| = o(1).$$
Let $\epsilon > 0$ and pick $m$ large enough so that
$$\|\langle y_{m+k+1} \rangle -  f_{m+k}\circ\cdot\cdot\circ f_{k}( \langle y_{k} \rangle )\| < \frac{\epsilon}{3}$$
and from (\ref{f-convg}),
$$\|f_{m+k}\circ\cdot\cdot\circ f_{k}(y) - P_\mathcal{X}(y)\| < \frac{\epsilon}{3}$$
for all $y \in$ a closed ball containing $\{\langle y_k \rangle\}$ along the chosen subsequence. Along the same subsequence, choose $k$ large enough so that
$$\|\langle y_k \rangle - \langle y^* \rangle \| < \frac{\epsilon}{3}.$$
Combining and using non-expansivity of projection, we have

\begin{align*}
\|\langle y_{m+k+1} \rangle - P_{\mathcal{X}}(\langle y^* \rangle) \| & \leq \|\langle y_{m+k+1} \rangle - f_{m+k}\circ\cdot\cdot\circ f_{k}( \langle y_{k} \rangle )\|+\| f_{m+k}\circ\cdot\cdot\circ f_{k}( \langle y_{k} \rangle )- P_{\mathcal{X}}(\langle y_k\rangle)\|\\
& \qquad  \qquad  \qquad  \qquad + \ \|P_{\mathcal{X}}(\langle y_k \rangle) - P_{\mathcal{X}}(\langle y^* \rangle) \|\\
& <\frac{\epsilon}{3}+ \frac{\epsilon}{3}+\frac{\epsilon}{3} =\epsilon
\end{align*}
The claim follows.


\end{proof}

\noindent \textit{B. Convergence of DSA-BDH }: The convergence analysis is similar to the previous case and Theorem \ref{main thrm} also holds for DSA-BDH. With a stacked
vector notation, the main steps of the algorithm are :
\begin{eqnarray}
x_{k} &=& (Q\varotimes I_{n})\{x_{k-1}+\mathbf{P}(z_{k}+y_{k})\}-\mathbf{P}(z_{k}+y_{k}), \label{fast3} \\
z_{k+1} &=& z_{k}+b_{k}x_{k}, \label{medium3} \\
y_{k+1} &=& (Q\varotimes I_{n})y_{k}+a_{k}(\mathbf{P}(y_{k}+z_{k})-y_{k})+a_{k}(h(y_{k})+M_{k+1}). \label{slow3}
\end{eqnarray}

The consensus part for $ y_k $ is proved along the same lines as Lemma \ref{consensus}. All that is required to show is that $\| \mathbf{P}(y_{k}+z_{k})-P_{\mathcal{X}} (y_{k}) \| \to 0 $, then the proof of Theorem \ref{main thrm} goes through with  only slight modification. (Here $\mathbf{P}(y_{k}+z_{k})$ performs the same job as $z_k$ of DSA-GD, i.e., asymptotically track the projection of $y_k$.) We proceed to show this now with the help of the following lemma proved in the appendix.

\begin{lem} \label{xbound}
$\|x_{k}\|$ is bounded and
\[
\|x_{k}-\{(Q^{*}\varotimes I_{n})\mathbf{P}(z_{k}+y_k)-\mathbf{P}(z_{k}+y_k)\}\|\to 0.
\]

\end{lem}
Using Lemma \ref{xbound} let us
write (\ref{medium3}) as

\begin{equation}\label{b-k-z}
z_{k+1}  = z_{k}+b_{k}\big( (Q^{*}\varotimes I_{n}) \mathbf{P}(z_{k}+y_{k})-\mathbf{P}(z_{k}+y_{k}) + o(1) \big)
\end{equation}
Since $a_{k}=o(b_{k})$ in (\ref{slow3})-(\ref{b-k-z}), they constitute a two time scale iteration (\cite{Borkar2time} or \cite{BorkarBook}, Chapter 6). So while analyzing (\ref{medium3}), we can assume that $y_{k}$ is $\approx$ a constant, say $\mathbf{y} := [y,....,y] $.
Add $\textbf{y}$ to both sides  of (\ref{medium3}) to obtain
\begin{eqnarray}
z_{k+1} + \textbf{y} & =& z_{k} + \textbf{y} + b_{k}x_{k} \nonumber \\
\mbox{i.e.,} \ \ r_{k+1} & =& r_{k}+b_{k}x_{k} \label{r-it}
\end{eqnarray}
with $r_{k}=z_{k}+\textbf{y}$. Thus  (\ref{fast3}), (\ref{r-it}) can be written as
\begin{eqnarray}
x_{k}&=&(Q\varotimes I_{n})\{x_{k-1} + \mathbf{P}(r_{k})\}-\mathbf{P}(r_{k}), \label{fast4} \\
r_{k+1}&=&r_{k} + b_{k}x_{k}. \label{medium4}
\end{eqnarray}
This is exactly the distributed Boyle-Dykstra-Han projection
algorithm in the variable $r_{k}$ (see (\ref{exeq})-(\ref{zedeq})).

\begin{lem}\label{mainlem-soham}
If $r^{*}$ is any limit point of (\ref{medium4}) (or equivalently,
(\ref{medium3})) with $r^{*}=z^{*}(y)+\mathbf{y} $ for some $z^{*}$ and a fixed
$\mathbf{y} $, then
\[
\mathbf{P}(r_{k})\to \mathbf{P}(r^{*})=P_{\mathcal{X}}(y).
\]

\end{lem}
\begin{proof}
Since $r_{k}=z_{k}+\mathbf{y} $ and $z_{0}=0$, we have $r_{0}=\mathbf{y} $. We first show that $r_{k}$
remains invariant under averaging.  From (\ref{fast3}),
 $$(Q^{*}\varotimes I_{n})x_{k+1} = (Q^{*}\varotimes I_{n})x_{k}=\cdots = (Q^{*}\varotimes I_{n})x_{0} = 0.$$
Multiply both sides of (\ref{medium4}) by $Q^{*}\varotimes I_{n}$
to get
 :
\[
(Q^{*}\varotimes I_{n})r_{k+1}=(Q^{*}\varotimes I_{n})r_{k}.
\]
By iterating we get $(Q^{*}\varotimes I_{n})r_{k}=(Q^{*}\varotimes I_{n})r_{0}$.
Since $z_{0}^{i}=0 \ \forall i$, we have
\[
(Q^{*}\varotimes I_{n})r_{k}=(Q^{*}\varotimes I_{n})(z_{k}+\mathbf{y} )=(Q^{*}\varotimes I_{n})(z_{0}+\mathbf{y} )=\mathbf{y}.
\]
That is, $\frac{1}{N}\sum_{i}r_{k}^{i}$ remains a constant equal to $y$ and since $r^*$ is a limit point of $\{r_n\}$, $\frac{1}{N}\sum_{i}r^*_{i} = y$.
Furthermore, we have the existence of a point $c \in \mathbb{R}^n$ such
that (see Lemma 4.4, \cite{Phade})
\[
\mathbf{P}(r^{*})=[P^{1}(r_{1}^{*})\,,\cdot\cdot\cdot,\,P^{N}(r_{N}^{*})]=[c\,,\cdot\cdot\cdot,c]
\]
\[
\frac{1}{N}\sum_{i}r_{i}^{*}=\frac{1}{N}\sum_{i}r_{k}^{i}=y,
\]
Since $P^i(r_i^*)=c \, \ \forall i$ we have $c \in \mathcal{X}=\cap_i \mathcal{X}_i$. This in turn implies for all $i$, $P_\mathcal{X}(r^*_i)=c$ because $\mathcal{X} \subset \mathcal{X}_i $.
Hence for all $i$, $r_i^*$ lie in the normal cone at the point $c$. Hence so does  $y$. This means that $c =P_\mathcal{X}(y)$ which proves the claim.

\end{proof}

For the original algorithm, this translates into:

\begin{cor} $ \| \mathbf{P}(r_k)-P_{\mathcal{X}} (y_{k}) \| \to 0 $
\end{cor}

As in Lemma \ref{projlem1}, this is a direct consequence of Lemma 1, Chapter 6, \cite{BorkarBook}.
The rest of the analysis is exactly the same as for DSA-GD and is therefore omitted.

\section{Stability}

In this section we give a sufficient condition for  the proposed algorithms to satisfy Assumption 2(iv) (i.e., have a.s.\ bounded iterates). \\

\noindent \textit{A. Stability of DSA-GD :} Boundedness of $z_k$ is obvious because it is projected onto a compact set at every step. To establish stability of $y_k$ we adapt a stability test from  \cite{BorkarBook}, Chapter 3, which was originally proposed in \cite{BorkarMeyn}.\\

Let $h_{c}^{i}(y^{i}):= \frac{h^{i}(cy^{i})}{c}$. Consider the following scaling limit for each $i$, assumed to exist:
\[
h_{\infty}^{i}(y^{i}):=\lim_{c\to\infty}\frac{h^{i}(cy^{i})}{c}.
\]

Note that the $h^i_c$'s have a common Lipschitz constant and hence are equicontinuous, implying that the above convergence is uniform on compact sets.
Suppose for each $i$ the following limiting ODE has the origin as
the unique globally asymptotically stable equilibrium :

\[
\dot{y}^{i}(t) =h_{\infty}^{i}(y^{i}(t))-y^{i}(t).
\]

Then if we let $H_{c}(y) := \frac{1}{N}\sum_{i=1}^{N}h_{c}^{i}(y^{i})=\frac{1}{N}\sum_{i=1}^{N}\frac{h^{i}(cy^{i})}{c},\,c\geq1,\,y\in\mathbb{R}^{nN}\,$,
it will satisfy
\[
H_{c}(y)\to H_{\infty}(y):=\frac{1}{N}\sum_{i=1}^{N}h_{\infty}^{i}(y^{i})\:\text{as }c\to\infty
\]
uniformly on compacts. Consider the following limiting ODE in $y(t)=[y^{1}(t), ..., y^{N}(t)]$,
which will have the origin ($\mathbf{0}\in\mathbb{R}^{nN}$) as the
unique globally asymptotically stable equilibrium :
\begin{align*}
\dot{y}(t) & =h_{\infty}(y(t))-y(t)\\
\Longrightarrow (\mathbf{1}^{T}\otimes I_{n})\dot{y}(t) & =(\mathbf{1}^{T}\otimes I_{n})\{h_{\infty}(y(t))-y
(t)\}\\
\Longrightarrow \frac{1}{N}\sum_{i=1}^{N}\dot{y}^{i}(t) & =\frac{1}{N}\sum_{i=1}^{N}\{h_{\infty}^{i}(y^{i}(t))-y^{i}(t)\}.
\end{align*}

If there is consensus so that $y^{i}(t)=y^{j}(t)=\langle y(t) \rangle $, we can write
the above as
\begin{equation} \label{stabeq}
\langle \dot{y}(t) \rangle=H_{\infty}(\langle y(t) \rangle)-\langle y(t) \rangle.
\end{equation}

Again, (\ref{stabeq}) has the origin ($\mathbf{0}\in\mathbb{R}^{n}$) as the
unique globally asymptotically stable equilibrium.\\
Let $T_{0}=0$ and for $n \geq 0$, $T_{n+1}=\min\{t_{m}:t_{m}>T_{n}+T\}$ where $T>0$
and $t_m = \sum_{k=0}^ma_k $. Without loss of generality, let $\sup_{k}a_{k}\leq1$. Then we have $T_{n+1}\in[T_{n}+T,T_{n}+T+1]\, \ \forall n$. Also,
we write $T_{n}=t_{m(n)}$ for a suitable $m(n)$. Let $y(t)=[y^{1}(t),...,y^{N}(t)]$
define a continuous, piecewise linear trajectory linearly interpolated between $y(t_k) := y_k$ and $y(t_{k+1}) := y_{k+1}$ on $[t_{k},t_{k+1}]$.\\

We construct another piecewise linear  trajectory $\hat{y}(t)$
derived from $y(t)$ by setting $\hat{y}(t)=\frac{y(t)}{r_{n}}$
for $t\in[T_{n},T_{n+1})$, where $r_{n}=\max_{i=1,...,N}\{\|y^{i}(T_{n})\|,1\}$.
This implies in particular that
\begin{equation} \label{ybound}
\hat{y}^{i}(T_{n})\leq1\: \ \forall\,n,i.
\end{equation}
Let $y^{\infty}(t)$ denote a generic solution to the equation (\ref{stabeq}) and
$y_{n}^{\infty}(t)$ denote its solution which starts at the point
$\langle\hat{y}(T_{n})\rangle$. For later use, we let $\hat{y}(T_{n+1}^{-})=y(T_{n+1})/r_{n}$. The following lemma is from \cite{BorkarBook} :

\begin{lem}\label{lemtemp}
{(i) $\exists\,T>0$ such that
for all initial conditions $y$ on the unit sphere, $\|y^{\infty}(t)\|<\frac{1}{8}$
for all $t>T$.}

\textit{(ii) The sequence $\xi_{k}\doteq\sum_{p=0}^{k-1}a_{p}\hat{M}_{p+1},\,k\geq1$
with $\hat{M}_{\ell}=\frac{M_{\ell}}{r_{n}}$ for $m(n) \leq \ell < m(n+1)$ is square integrable and a.s.\ convergent }

\textit{(iii) The sequence $\hat{\xi}_{k}\doteq\sum_{p=0}^{k-1}a_{p} A_p \hat{M}_{p+1}, k\geq1$
with $\{A_p\}_{p\geq 0}$ being a sequence of doubly stochastic matrices is a.s.\ convergent }

\end{lem}

\begin{proof}  The proof of (i) is the same as in Chapter 3, Lemma 1, \cite{BorkarBook}, pp.\ 22-23, whereas (ii) is proved in Chapter 3, Lemma 5, \cite{BorkarBook}, p.\ 25.  The latter lemma establishes and uses the fact
$$\sum_{p=0}^{\infty}a_p^2\mathbb{E}\left[\|\hat{M}_{p+1}\|^2 | \mathcal{F}_p\right] < \infty.$$
To prove (iii), note that multiplication by a linear operator doesn't affect the martingale property. Also,
  $$ \sum_{p=0}^{\infty}a_p^2\mathbb{E}\left[\| A_p \hat{M}_{p+1}\|^2 | \mathcal{F}_p\right] \leq \sum_{p=0}^{\infty}a_p^2\mathbb{E}\left[\|\hat{M}_{p+1}\|^2 | \mathcal{F}_p\right] < \infty\,\,\text{a.s.}$$
Convergence follows from  the martingale convergence theorem (Appendix C, \cite{BorkarBook}).

\end{proof}

The following theorem proves that the iterates $y_k$ remain bounded. Since the proof is an adaptation of the arguments of \cite{BorkarMeyn} or \cite{BorkarBook}, Chapter 2, we give only a sketch that highlights the significant points of departure.

\begin{thm} \label{ystab}
$\sup_k\|y_k\|<\infty$. \end{thm}

\begin{proof}
(Sketch) Suppose that $\| y^{i}(t)\|\to\infty$ for some $i$ along a subsequence.  We obtain a contradiction using the following argument :\\

\noindent \textit{Claim :} $\lim_{n\to\infty}\sup_{t\in[T_{n},T_{n+1}]}\|\langle\hat{y}(t)\rangle-y_{n}^{\infty}(t)\|=0$
a.s.\\

\noindent \textit{Proof.} For $m(n)<k<m(n+1)$, we have, on dividing both sides of (\ref{slow}) by $r_n$,
\begin{equation}\label{alt}
\hat{y}(t_{k+1})=(Q\varotimes I_{n})\hat{y}(t_{k})+a_{k}(h_{r_{n}}(\hat{y}(t_{k}))-\hat{y}_{k}(t)+\hat{M}_{k+1}+\epsilon_{k}),
\end{equation}
where $\epsilon_{k}^{i}=\frac{z_k^i}{r_{n}}$.
Since $z_k^i$ is bounded, $\epsilon_{k}^{i}\to0\, \ \forall i$ as $n\to\infty$ (since $r_{n}\to\infty$
by assumption). Iterating the above equation we get,

\begin{multline*}
\hat{y}(t_{m(n)+k})=(Q^{k}\varotimes I_{n})\hat{y}(t_{m(n)})+ \sum_{i=0}^{k-1} a_{m(n)+i} (Q^{k-i-1}\varotimes I_{n}) \big\{ h_{r_{n}}(\hat{y}(t_{m(n)+i}))-\hat{y}(t_{m(n)+i})\\ +\hat{M}_{m(n)+i+1}+\epsilon_{m(n)+i}\big\}.
\end{multline*}

Taking norms in the above we have :
\begin{multline}\label{newno}
\| \hat{y}(t_{m(n)+k})\| \leq \| \hat{y}(t_{m(n)})\| + \sum_{i=0}^{k-1} a_{m(n)+i} \big\{ \| h_{r_{n}}(\hat{y}(t_{m(n)+i}))-\hat{y}(t_{m(n)+i}) +\epsilon_{m(n)+i}\|\big\} \\ + \| \sum_{i=0}^{k-1} a_{m(n)+i} (Q^{k-i-1}\varotimes I_{n})\hat{M}_{m(n)+i+1} \|.
\end{multline}
Let $L$ be a common Lipschitz
constant for the functions $h_{r_{n}}(\cdot)$ (which in fact is the same
as that for $h$). Then we have the following bound on $h_{r_n}(\cdot)$ :
\[
\|h_{r_{n}}(\hat{y}(t_{k}))\|\leq\|h_{r_n}(0)\|+L\|\hat{y}(t_{k})\|
\]
Using the above in (\ref{newno}), we have :
\begin{multline*}
\| \hat{y}(t_{m(n)+k})\| \leq \| \hat{y}(t_{m(n)})\| + \sum_{i=0}^{k-1} a_{m(n)+i} \big\{ (L+1) \| \hat{y}(t_{m(n)+i})\| +\| \epsilon_{m(n)+i}\| + \| h_{r_n}(0) \| \big\}\\ +  \| \sum_{i=0}^{k-1} a_{m(n)+i} (Q^{k-i-1}\varotimes I_{n})\hat{M}_{m(n)+i+1} \|.
\end{multline*}

We first prove that  $\sup_k \| \sum_{i=0}^{k-1} a_{m(n)+i} (Q^{k-i-1}\varotimes I_{n})\hat{M}_{m(n)+i+1} \| < \infty$. Define the sequence  $\{ \hat{\xi}_p \}$ as :
\[
\hat{\xi}_{p} = \hat{\xi}_{m(n)} + \sum_{i=0}^{p-m(n)-1} a_{m(n)+i} (Q^{p-m(n)-i-1}\varotimes I_{n})\hat{M}_{m(n)+i+1} \,\,\,\,\, \text{ if } m(n) < p \leq m(n+1)
\]

with $\hat{\xi}_{0}=0 $ and

\[
\hat{\xi}_{m(n)} = \sum_{j=0}^{n-1} \Big( \sum_{i=0}^{m(j+1) - m(j)-1} a_{m(j)+i} (Q^{m(j+1)-m(j)-i-1}\varotimes I_{n})\hat{M}_{m(j)+i+1} \Big)
\]

The sequence $\{ \hat{\xi}_p \}$ is convergent by Lemma \ref{lemtemp}(iii). Then  $$ \| \sum_{i=0}^{k-1} a_{m(n)+i} (Q^{k-i-1}\varotimes I_{n})\hat{M}_{m(n)+i+1}\| = \| \hat{\xi}_{m(n)+k} -\hat{\xi}_{m(n)}\| := B_n, $$ where $B_n\to 0$ a.s.  Also $\sum_{0\leq p\leq m(n+1)-m(n)}a_{m(n)+p}\leq T+1$, so
that
\[
\|\hat{y}(t_{m(n)+k})\| \leq \|\hat{y}(t_{m(n)})\| +C+(L+1)\sum_{i=0}^{k-1}a_{m(n)+i}\|\hat{y}(t_{m(n)+p})\|,
\]
where $C\geq(T+1) \|h(0)\| + B_n +\sum_{i=0}^{m(n+1)-m(n)}a_{n+i} \|\epsilon_{k+i}\|$ is a random constant. The last term is finite a.s.\  because the $\epsilon_{k}$ term goes to zero as stated earlier.
Since $\|\hat{y}(t_{m(n)})\|\leq1$ by (\ref{ybound}), we have
\[
\|\hat{y}(t_{m(n)+k})\| \leq (C+1)+(L+1)\sum_{i=0}^{k-1}a_{m(n)+i}\|\hat{y}(t_{m(n)+p})\|,
\]
By discrete Gronwall inequality, we have
\begin{equation}\label{yhatbound}
\sup_{0 \leq k \leq m(n+1) - m(n)}\|\hat{y}(t_{m(n)+k})\| \leq [C+1]\exp\{(L+1)(T+1)\}\equiv K^*.
\end{equation}
Since the bound is independent of $n$, we have that
\begin{equation} \label{boundedy}
\|\hat{y}(t_{m(n)+k})\|<\infty.
\end{equation}
Now consider (\ref{alt}) again :
\begin{equation*}
\hat{y}(t_{k+1})=(Q\varotimes I_{n})\hat{y}(t_{k})+a_{k}(h_{r_{n}}(\hat{y}(t_{k}))-\hat{y}_{k}(t)+\hat{M}_{k+1}+\epsilon_{k}).
\end{equation*}

Multiply it on both sides by $\frac{1}{N}(\mathbf{1}^{T}\otimes I_{n})$ to
get
\begin{eqnarray}
\langle\hat{y}(t_{k+1})\rangle&=&\langle\hat{y}(t_{k})\rangle+a_{k}\{\frac{1}{N}\sum_{i=1}^{N}h_{r_{n}}^{i}(\hat{y}^i(t_{k}))-\langle\hat{y}(t_{k})\rangle+\langle\hat{M}_{k+1}\rangle+\langle \epsilon_{k} \rangle \} \nonumber \\
&=& \langle\hat{y}(t_{k})\rangle+a_{k}\{H_{\infty}(\langle\hat{y}(t_{k})\rangle)-\langle\hat{y}(t_{k})\rangle+\delta_{k}^{1}+\delta_{k}^{2}+\langle\hat{M}_{k+1}\rangle+\langle \epsilon_{k} \rangle \}. \label{finaleq}
\end{eqnarray}
where
\begin{enumerate}
 \item $\delta_{k}^{1}=\frac{1}{N}\sum_{i=1}^{N}h_{r_{n}}^{i}(\hat{y}^i(t_{k}))-\frac{1}{N}\sum_{i=1}^{N}h_{r_{n}}^{i}(\langle\hat{y}^i(t_{k})\rangle)$. Since $\{\langle\hat{y}_k\rangle\}$ is bounded by (\ref{boundedy}) we can adapt the arguments of  Lemma \ref{consensus} of Section 3 to show that we achieve consensus and hence $\|\delta_{k}^{1}\|\to0$ : For any $m(n)<k<m(n+1)$ we have from (\ref{alt}) :

 $$\hat{y}(t_{m(n)+k})=(Q\varotimes I_{n})\hat{y}(t_{m(n)+k-1})+a_{m(n)+k-1}\hat{Y}_{m(n)+k-1},$$

 where $$\hat{Y}_{m(n)+k-1} = h_{r_{n}}(\hat{y}(t_{m(n)+k-1}))-\hat{y}(t_{m(n)+k-1})+\hat{M}_{m(n)+k}+\epsilon_{m(n)+k-1}.$$

 Iterating this equation we get,
$$\hat{y}(t_{m(n)+k})=(Q^{k}\varotimes I_{n})\hat{y}(t_{m(n)})+\{\Gamma(\hat{Y}_{m(n)+k-1},...,\hat{Y}_{m(n)})\}. $$

 where $\Gamma(\cdot)$ is defined as in Lemma \ref{consensus} and is of order $\mathcal{O}(\sum_{i=m(n)}^{m(n)+k}a_i)$ because $\hat{y}(\cdot)$ (and hence $\hat{Y}$) is bounded. Then as in Lemma \ref{consensus},  we have,

 \[
\|\hat{y}(t_{m(n)+k})-(Q^{*}\varotimes I_{n})\hat{y}(t_{m(n)})\| =\mathcal{O}(\beta^{-k})+\mathcal{O}(\sum_{i=m(n)}^{m(n)+k-1}a_i),
\]
where the $\mathcal{O}(\beta^{-k})$ term is uniform w.r.t.\ $n$.
Taking the limit $n \to \infty$ (which means $m(n) \to \infty$) followed by $k\to \infty$, we get that $\|\hat{y}^i(\cdot)-\langle\hat{y}^i(\cdot)\rangle \| \to 0$ and hence $\| \delta_k^1 \| \to 0$.

\item $\delta_{k}^{2}=\frac{1}{N}\sum_{i=1}^{N}h_{r_{n}}^{i}(\langle\hat{y}^i(t_{k})\rangle)-\frac{1}{N}\sum_{i=1}^{N}h_{\infty}^{i}(\langle\hat{y}^i(t_{k})\rangle)$.
Since $r_{n}\to\infty$ , we have by assumption $H_{r_{n}}\to H_{\infty}$ uniformly on compact sets, so
 that by (\ref{boundedy}), $\|\delta_{k}^{2}\|\to0$.
\end{enumerate}
Now we can use standard arguments, e.g., of  Lemma 1 and Theorem 2, Chapter 2, of \cite{BorkarBook}, pp.\ 12-15,  which
show that (\ref{finaleq}) has the same asymptotic behavior as  the ODE (\ref{stabeq}) a.s., which proves the claim.\\

The above claim  gives the contradiction we require. Suppose without loss of generality that $\|y(T_n)\|>1$ along the above subsequence, which we denote by $\{n\}$ again by abuse of notation. Then $r_n\to\infty$ and we have a sequence $T_{n_1},T_{n_2}...$ such that $\|y(T_{n_k})\|\uparrow \infty$, i.e., $r_{n_k} \uparrow \infty$.  We have $\|\langle\hat{y}(T_{n})\rangle\|=\|y_{n}^{\infty}(T_{n})\|\leq1$ and by Lemma \ref{lemtemp}(i) we may take $\|y_{n}^{\infty}(T_{n+1})\|<\frac{1}{8}$ since
$T_{n+1}>T+T_{n}$. So  by the above claim there exists an $N'$ such that
for all $n>N'$, we get $\|\langle\hat{y}(T_{n+1}^{-})\rangle\|<\frac{1}{4}$.
Then for all sufficiently large $n$,
\[
\frac{\|\langle y(T_{n+1})\rangle\|}{\|\langle y(T_{n})\rangle\|}=\frac{\|\langle \hat{y}(T_{n+1}^{-})\rangle\|}{\|\langle\hat{y}(T_{n})\rangle\|}<\frac{1}{4}.
\]
We conclude that if $\|y(T_n)\|>1$, $\|y(T_k)\|$ for $k \geq n$ falls back to the unit ball at an exponential rate. Thus if $\|y(T_n)\|>1$, $\|y(T_{n-1})\|$ is either even greater than  $\|y(T_{n})\|$ or inside the unit ball. Thus there must an instance prior to $n$ when $y(\cdot)$ jumps from inside the unit ball to a radius of $0.9\,r_n$. Then we have a sequence of jumps of $y(T_n)$, corresponding to the sequence $r_{n_k}\to\infty$, from inside the unit ball to points increasingly far away from the origin. But, by a discrete Gronwall argument analogous to the one used in the above claim, it follows  that there is a bound on the amount by which $y(\cdot)$ can increase over an interval of length $T+1$ when it is inside the unit ball at the beginning of the interval. This leads to a contradiction, implying $\tilde{C}=\sup_n\|y(T_n)\| <\infty $. This implies by  (\ref{yhatbound}) that
$\sup_n\|y_n\|\leq \tilde{C}K^* <\infty $.
\end{proof}

\noindent \textit{B. Stability of DSA-BDH } :  The proof that $y_k$ remains stable for DSA-BDH is identical to Theorem \ref{ystab}. The stability of $z_k$ and $x_k$ can be handled as in \cite{Phade} by minor modifications of the arguments therein. The proof uses a routine ODE approximation technique and we skip it here as it is quite lengthy.

\section{Numerical Experiment}

In this section, we present some numerical results to validate the proposed algorithms. We demonstrate the results on a stochastic optimization problem known as the stochastic utility problem.  We consider this problem for a max function with linear arguments as the objective and the unit simplex as the constraint set (Section 4.2, \cite{Nemirovski})  :

$$ \min_{y\in\mathcal{X}} \Big\{ \mathbb{E}[F(y,\xi)] = \mathbb{E}
\Big[\phi \big(\sum_{i=1}^n ( \frac{i}{n}+\xi(i) )y(i)\big)\Big]\Big\}$$
$$\mathcal{X} =  \{y\in \mathbb{R}^n\,:\,y(i)\geq 0 \,\forall i\,,\sum_{i=1}^ny(i)=1 \},\,\,\xi(i)\in \mathcal{N}(0,1) $$
where $y=[y(1),....,y(n)]$ and $\phi(t)=\max\{v_1+s_1t,.....,v_m+s_mt\}$ with $v_k$ and $s_k$ being constants. Also, $\xi(i)$ are independent zero mean Gaussian random variables with unit standard deviation. The constants $v_k$ and $s_k$ are generated from a uniform distribution with $m=10$.
         We test the DSA-GD algorithm with this setup for $N=10,20$ and $30$, where $N$ denotes the number of constraints (equal to $n+1$) and hence the number of nodes we require to do the projection. If we put the constraints in the intersection form $\bigcap_i \mathcal{X}_i $, then $\mathcal{X}_i $ will be :

        \begin{equation*}
  \mathcal{X}_i=\begin{cases}
    \{ y \in  \mathbb{R}^n : y(i) \geq 0 \}, & \text{for } 1\leq i \leq N-1.\\
   \{ y \in \mathbb{R}^n  : \sum_{i=1}^N y(i)  \geq 1\}, &  \text{for }  i=N.
  \end{cases}
\end{equation*}
Node $i$ is assigned the constraint set $\mathcal{X}_i$ in order to do a distributed projection. The matrix $Q$ is generated using Metropolis weights \footnote{The code to generate it is borrowed from \cite{Sim} } and for $N=10$ is explicitly given as  :
$$Q= \left[\begin{array}{cccccccccc}
0.25 & 0.25 & 0& 0 & 0 & 0 & 0 & 0.25 & 0 &  0.25\\
0.25 & 0.4167 & 0.333 & 0 & 0 & 0 & 0 & 0 & 0 & 0 \\
0 & 0.333 & 0.333 & 0.333 & 0 & 0 & 0 & 0 & 0 & 0\\
 0 & 0 & 0.333 & 0.333 & 0.333 & 0 & 0 & 0 & 0 & 0\\
0 & 0 & 0 & 0.333 & 0.333 & 0.333 & 0 & 0 & 0 & 0\\
0 & 0 & 0 & 0 & 0.333 & 0.333 & 0.333 & 0 & 0 & 0\\
0 & 0 & 0 & 0 & 0 & 0.333 & 0.4167 & 0.25 & 0 & 0\\
0.25 & 0 & 0 & 0 & 0 & 0 & 0.25 & 0.25 & 0.25 & 0\\
0  & 0 & 0 & 0 & 0 & 0 & 0 & 0.25 & 0.4167 & 0.333\\
0.25 & 0 & 0 & 0 & 0 & 0 & 0 & 0  & 0.333 & 0.4167
\end{array}\right]$$
The above matrix is consistent with Assumption 1. The time steps employed are $b_k=\frac{1}{k^{0.7}}$ and $a_k = \frac{1}{k^{0.95}}$. Also, $h^i(y^i,\xi) = \partial F (y^i,\xi)$, where $\partial F $ denotes the sub-gradient. Being a sub-gradient descent, the problem does not satisfy the regularity hypotheses imposed in our analysis above, nevertheless the proposed schemes work well as we show  below.\\

We plot our results in Figures 1,2 and 3 :
\begin{itemize}

\item Figure 1 shows the plot of the optimality error vs. the number of iterations. The optimality error is the difference $\|y^1_k -y^{1,*}\|$, where $y^{1,*}$ is the output after running the algorithm long enough ($ k > 10^4$) and the error tolerance $\frac{\|y^1_{k+1} -y^1_k \|}{\|y_k\|}$ is sufficiently small. As expected, the number of iterations required increases with the dimension of the problem.

\item Figure 2 shows the feasibility error, $\|y^1_k - P_\mathcal{X}(y^1_k)\|$, against the iteration count.

\item Figure 3 shows the disagreement estimate, $\|y^i_k - y^j_j\|$, between the various agents for $i,j =1,2,3,4$.

\end{itemize}

\begin{figure}[H]
\begin{center}
\includegraphics[width=11cm,height=8cm]{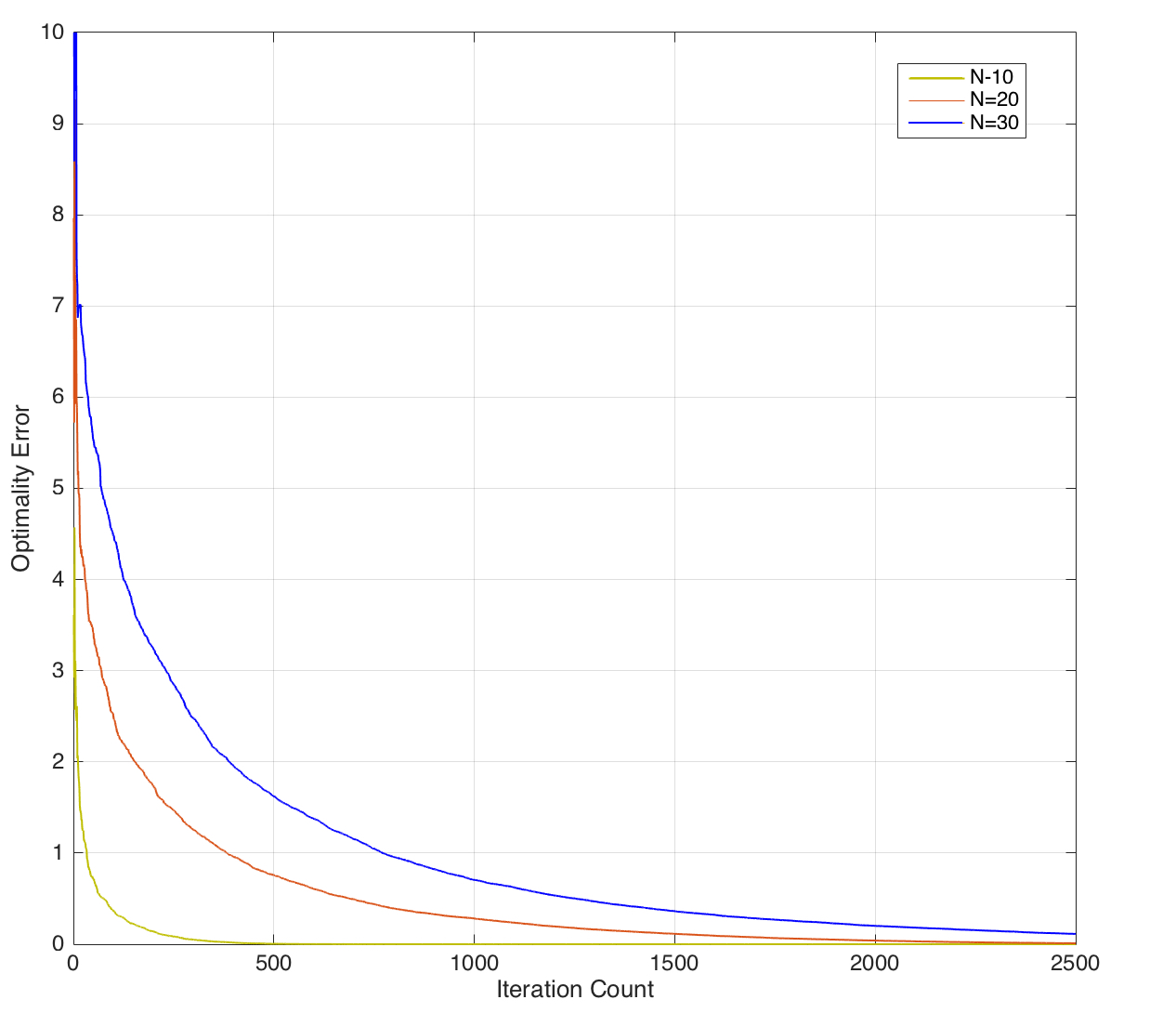}
\end{center}
\caption{Optimality Error vs. Iteration Count}

\end{figure}

\begin{figure}[H]
\begin{center}
\includegraphics[width=11cm,height=8cm]{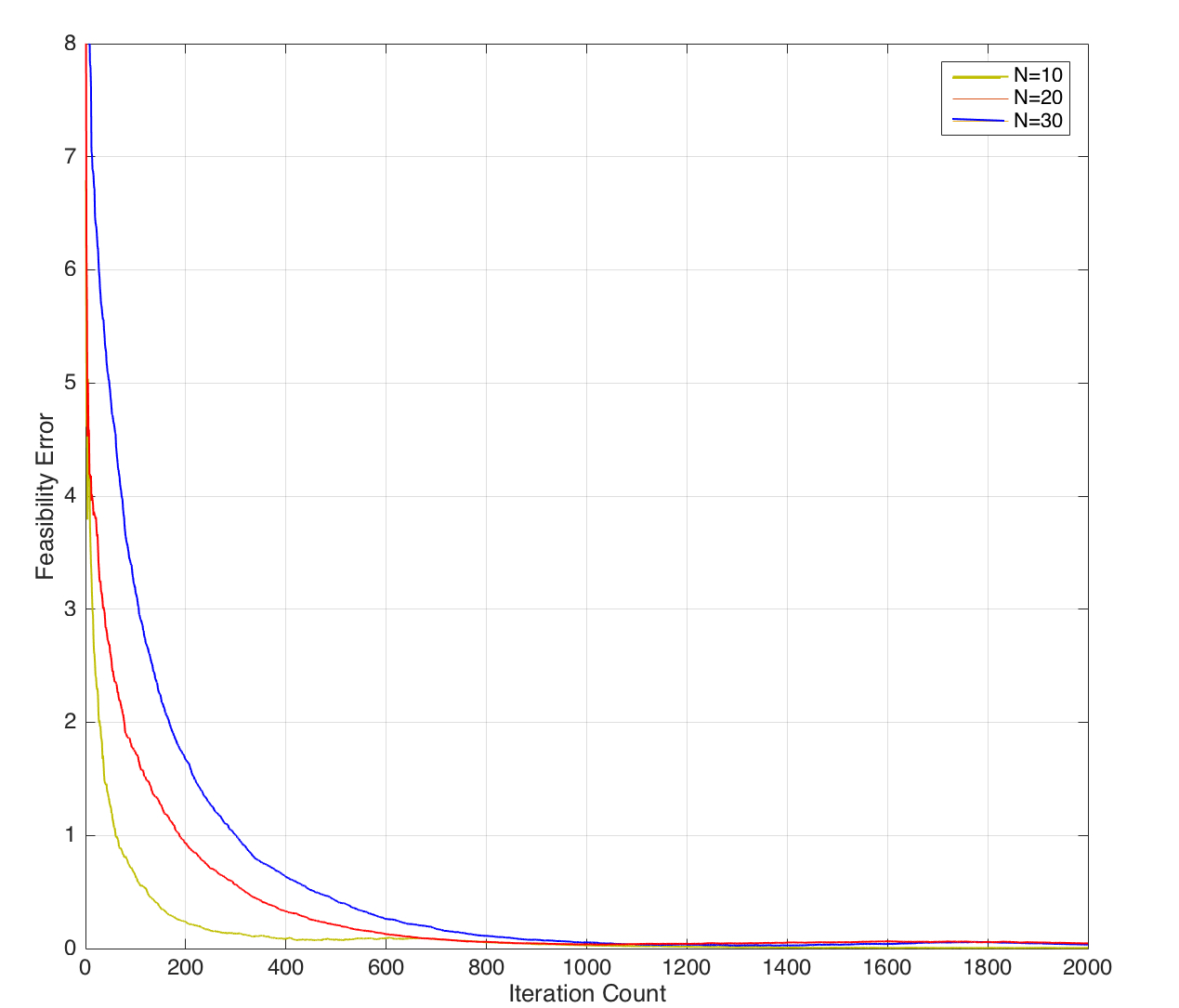}
\end{center}
\caption{Feasiblity Error vs. Iteration Count}
\end{figure}

\begin{figure}[H]
\begin{center}
\includegraphics[width=11cm,height=8cm]{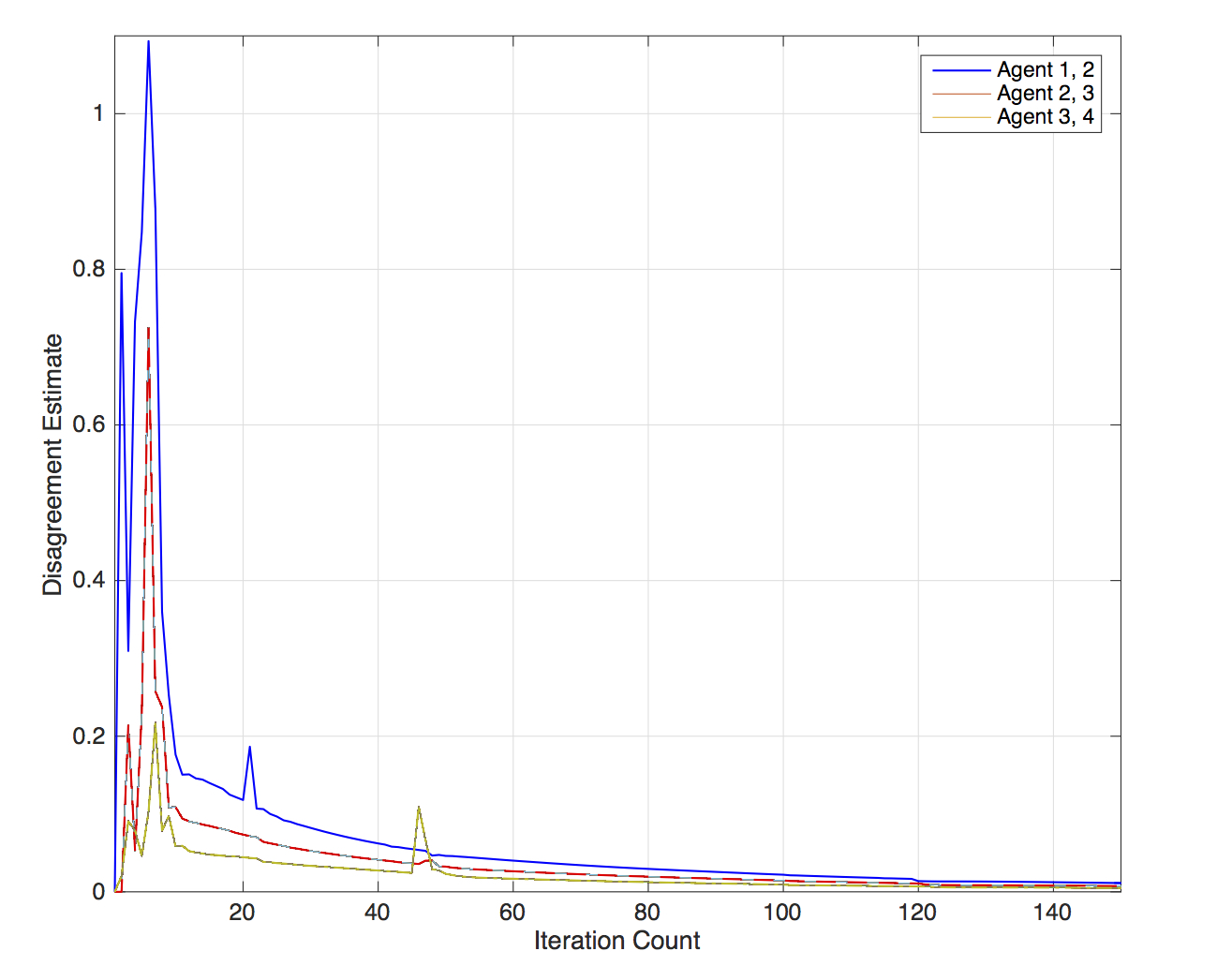}
\end{center}
\caption{Disagreement Estimate vs. Iteration Count}
\end{figure}

\section*{Appendix}
The proof of Lemma \ref{xbound} is along the same lines as in \cite{Phade} and we provide it here for sake of completeness. Recall the condition  (see eq. (\ref{apndxb1}))  on $b_k$ : For any $\epsilon>0$, there exists an $\alpha \in (1,1+\epsilon)$ and some $k_0$ such that
\begin{equation}\label{apndxb}
\alpha b_{k+1}\geq b_k,\,\,\,\forall k>k_0.
\end{equation}

\begin{proof}
For any $n$ and $k$ in (\ref{fast3}), we have
\[
x_{n+k}=(Q\varotimes I_{n})\{x_{n+k-1}+\mathbf{P}(z_{n+k}+y_{n+k})\}-\mathbf{P}(z_{n+k}+y_{n+k}).
\]

Setting $r_{k}=z_{k}+y_{k}$ and iterating the above equation we get,
\[
x_{n+k}=(Q^{k}\varotimes I_{n})x_{n}-\mathbf{P}(r_{n+k})+\sum_{i=1}^{k-1}(Q^{i}\varotimes I_{n})\{\mathbf{P}(r_{n+k-i+1})-\mathbf{P}(r_{n+k-i})\}+(Q^{k}\varotimes I_{n})\mathbf{P}(r_{n+1})
\]

Using the fact that $(Q^{*} \varotimes I_n )x_{n}=0$ and adding the telescopic sum  inside the curly brackets, we have
:
\begin{eqnarray*}
\lefteqn{x_{n+k} = ((Q^{k}-Q^{*})\varotimes I_{n})x_{n}-\mathbf{P}(r_{n+k})} \\
 &&+ \ \sum_{i=1}^{k-1}(Q^{i}\varotimes I_{n})\{\mathbf{P}(r_{n+k-i+1})-\mathbf{P}(r_{n+k-i})\}  + \ (Q^{k}\varotimes I_{n})\mathbf{P}(r_{n+1}) + \\
&&\Big\{-\sum_{i=1}^{k-1}(Q^{*}\varotimes I_{n})(\mathbf{P}(r_{n+k-i+1})-\mathbf{P}(r_{n+k-i}))+(Q^{^{*}}\varotimes I_{n})\mathbf{P}(r_{n+k})-(Q^{^{*}}\varotimes I_{n})\mathbf{P}(r_{n+1})\Big\}
\end{eqnarray*}
\begin{multline}\label{apndxmain}
\Longrightarrow x_{n+k}=((Q^{k}-Q^{*})\varotimes I_{n})(x_{n}+\mathbf{P}(r_{n+1})) +(Q^{*}\varotimes I_{n})\mathbf{P}(r_{n+k})-\mathbf{P}(r_{n+k})\\
+ \sum_{i=1}^{k-1}((Q^{i}-Q^{*})\varotimes I_{n})\{\mathbf{P}(r_{n+k-i+1})-\mathbf{P}(r_{n+k-i})\}.
\end{multline}

Taking $n=0$ in (\ref{apndxmain}) and using $\|\mathbf{P}(r)\|\leq C<\infty$ and (\ref{eq:-3}) to bound the norm of $(Q^{i}-Q^{*})\varotimes I_n$, we get
\begin{equation}\label{apndxstab}
\|x_{k}\|\leq \beta^{-k}\|x_{0}+\mathbf{P}(r_{1})\|+2C+2C\sum_{i=1}^{k-1}\kappa\beta^{i}.
\end{equation}

Since the RHS in the above is uniformly bounded, we have  $x_{k}$
uniformly bounded. Now consider (\ref{apndxmain}) again:
\begin{multline}\label{apndxfinal}
x_{n+k}-(Q^{*}\varotimes I_{n})\mathbf{P}(r_{n+k}) + \mathbf{P}(r_{n+k})=\underbrace{((Q^{k}-Q^{*})\varotimes I_{n})(x_{n}+\mathbf{P}(r_{n+1}))}_{(I)}\\
+\underbrace{\sum_{i=1}^{k-1}((Q^{i}-Q^{*})\varotimes I_{n})\{\mathbf{P}(r_{n+k-i+1})-\mathbf{P}(r_{n+k-i})\}}_{(II)}
\end{multline}

(I): This can be bounded by using (\ref{eq:-3}) and (\ref{apndxstab}):
\[
\|((Q^{k}-Q^{*})\varotimes I_{n})(x_{n}+\mathbf{P}(r_{n+1}))\|\leq\kappa\beta^{-k}\|x_{n}+\mathbf{P}(r_{n+1})\|=\mathcal{O}(\beta^{-k})
\]

(II): To bound this term, we first consider (\ref{medium3}) and add $y_{k+1}$ on both sides of
\[
z_{k+1}+y_{k+1}=z_{k}+y_{k}+b_{k}x_{k}+a_{k}(\mathbf{P}(r_{k})-y_{k}+h(y_{k})+M_{k+1}) + (Q\varotimes I_{n})y_k -y_k
\]
\[
\Longrightarrow \ r_{k+1}=r_{k}+b_{k}\{x_{k}+\frac{a_{k}}{b_{k}}(\mathbf{P}(r_{k})-y_{k}+h(y_{k})+M_{k+1})\}    + (Q\varotimes I_{n})y_k -y_k
\]
Since $y_{k}$ is bounded (cf. Assumption 2(iv)) and
$x_{k}$ is bounded from (\ref{apndxstab}),
\[
\|r_{k+1}- r_{k}\| \leq b_{k}M_{x,y} + \epsilon_k
\]
for some random constant $M_{x,y}<\infty$ a.s, with $\epsilon_k= \| (Q\varotimes I_{n})y_k -y_k\|$. We have, by assumption (\ref{apndxb}) on the step size $b_{k}$, that there exists
an $\alpha\in(1,\beta)$ and $k_{0}$ such that
\begin{equation}\label{balpha}
\alpha^{n+k-i}b_{n+k-i}\leq\alpha^{n+k}b_{n+k}\,\,\, \ \forall \ 1\leq i\leq k-1
\end{equation}

Letting $\hat{\beta}=\frac{\beta}{\alpha}>1$,

\begin{align*}
\|\sum_{i=1}^{k-1}((Q^{i}-Q^{*})\varotimes I_{n})\{\mathbf{P}(r_{n+k-i+1})-\mathbf{P}(r_{n+k-i})\}\| & \leq\sum_{i=1}^{k-1}\kappa\beta^{-i}\|\mathbf{P}(r_{n+k-i+1})-\mathbf{P}(r_{n+k-i})\|\\
 & \leq\sum_{i=1}^{k-1}\kappa\beta^{-i}\|r_{n+k-i+1}-r_{n+k-i}\|\\
 & \leq\sum_{i=1}^{k-1}\kappa\beta^{-i}\big(M_{x,y}b_{n+k-i} + \epsilon_{n+k-i}\big)\\
 & \stackrel{(\ref{balpha})}{\leq} \sum_{i=1}^{k-1}\kappa\hat{\beta}^{-i}M_{x,y}b_{n+k} + \sum_{i=1}^{k-1}\kappa\beta^{-i}\epsilon_{n+k-i} \\
 & \leq \frac{\kappa}{\hat{\beta}-1}M_{x,y}b_{n+k} + \frac{\kappa}{\beta-1}\bar{\epsilon}_{n,k}
\end{align*}
where $\bar{\epsilon}_{n,k} = \sup_{\{0 \leq i \leq k-1\}} \epsilon_{n+k-i} \to 0$ as $n \to \infty$. The lemma follows by substituting the above bounds on (I) and (II) in (\ref{apndxfinal}) and taking the limit $n \to \infty$ followed by $k\to \infty$ to obtain

$$  \|x_{n+k}-\{(Q^{*}\varotimes I_{n})\mathbf{P}(r_{n+k})-\mathbf{P}(r_{n+k})\}\| \to 0 $$
\end{proof}

 \end{document}